\documentclass[a4paper,11pt]{article}
\usepackage{fullpage}
\usepackage{graphics,amsmath,amssymb,amsthm,accents,epic}
\usepackage{mathrsfs}
\usepackage{cite}
\usepackage[top=1 in,bottom=1 in,left=1.0 in,right=1.0 in]{geometry}
\usepackage{bm}
\usepackage{subfigure}
\usepackage{latexsym}
\usepackage{xcolor}
\usepackage[colorlinks,linkcolor=black,anchorcolor=blue,citecolor=blue]{hyperref}
\usepackage{amsfonts}
\usepackage{mathtools}
\usepackage{multirow}
\usepackage{caption}
\usepackage{ulem}
\usepackage{titlesec}

\numberwithin{equation}{section}

\def \c#1{\accentset{\circ}{#1}}

\newtheorem{prop}{Proposition}[section]
\newtheorem{theorem}{Theorem}[section]

\newtheorem{remark}{Remark}[section]

\newcommand{\wh}{\widehat}
\newcommand{\wt}{\widetilde}

\newcommand{\bA}{\boldsymbol{A}}
\newcommand{\bB}{\boldsymbol{B}}
\newcommand{\bM}{\boldsymbol{M}}
\newcommand{\bL}{\boldsymbol{L}}
\newcommand{\br}{\boldsymbol{r}}

\newcommand{\ba}{\boldsymbol{a}}
\newcommand{\bC}{\boldsymbol{C}}
\newcommand{\bK}{\boldsymbol{K}}
\newcommand{\bS}{\boldsymbol{S}}
\newcommand{\bT}{\boldsymbol{T}}
\newcommand{\bV}{\boldsymbol{V}}
\newcommand{\bW}{\boldsymbol{W}}
\newcommand{\bu}{\boldsymbol{u}}
\newcommand{\bI}{\boldsymbol{I}}
\newcommand{\bv}{\boldsymbol{v}}
\newcommand{\bw}{\boldsymbol{w}}
\newcommand{\bG}{\boldsymbol{G}}

\newcommand{\be}{\boldsymbol{e}}
\newcommand{\bH}{\boldsymbol{H}}
\newcommand{\bJ}{\boldsymbol{J}}
\newcommand{\bF}{\boldsymbol{F}}

\newcommand{\bs}{\boldsymbol{s}}
\newcommand{\bP}{\boldsymbol{P}}
\newcommand{\bQ}{\boldsymbol{Q}}
\newcommand{\bp}{\boldsymbol{p}}
\newcommand{\bq}{\boldsymbol{q}}
\newcommand{\bU}{\boldsymbol{U}}
\newcommand{\bX}{\boldsymbol{X}}

\newcommand{\mV}{\mathcal{V}}

\newcommand{\tc}{\,^t \hskip -2pt {\boldsymbol{c}}}

\newcommand{\nn}{\nonumber}
\newcommand{\st}{\hbox{\tiny\it{T}}}

\title{Discrete Gerdjikov–Ivanov models and their higher-order counterparts from the Cauchy matrix scheme}	

\author{Xiao-gang Mu,~~Song-lin Zhao\footnote{
Corresponding Author: songlinzhao@zjut.edu.cn}\\
\small{School of Mathematical Sciences, Zhejiang University of Technology, Hangzhou 310023,  China}}

\begin{document}

\maketitle

\begin{abstract}

The Gerdjikov–Ivanov (GI) equation is an important model in the derivative nonlinear Schr\"{o}dinger system, 
yet its fully discrete integrable analogues remain unexplored. In this paper, we systematically construct 
discrete versions of both the GI equation and its higher-order counterpart (hGI equation) within the Cauchy matrix framework. 
Starting from the Sylvester equation equipped with two distinct sets of discrete dispersion relations, 
we derive the shift dynamics of the master functions and eliminate auxiliary variables to obtain closed lattice systems. 
Since the elimination step admits several equally valid algebraic identities, this procedure yields four conjugate-symmetric 
families of discrete GI (dGI) models and four families of discrete higher-order GI (dhGI) models. 
For each discrete model, we provide explicit $N$-soliton and multiple-pole solutions 
via the Cauchy matrix method with diagonal and Jordan-block spectral matrices, respectively. We verify through a two-step continuum limit, 
contracting one lattice direction at a time, that all four dGI models reduce to the same continuous GI equation and all four dhGI models 
reduce to the same continuous hGI equation. Finally, we investigate reductions: local complex conjugate reductions yield scalar dGI and dhGI 
equations with explicit solutions. Moreover, in the higher-order case, pairwise recombinations of the dhGI lattice equations admit nonlocal reductions 
that produce nonlocal dhGI equations and their solutions.

\begin{description}
\item[Keywords:] Discrete Gerdjikov–Ivanov models, Discrete higher-order Gerdjikov–Ivanov models, Cauchy matrix scheme, solutions, reduction
\end{description}

\end{abstract}

\section{Introduction}

The nonlinear Schr\"{o}dinger (NLS) equation constitutes one of the most important families 
of integrable systems in mathematical physics, with applications ranging from nonlinear optics and plasma physics to Bose--Einstein condensates 
\cite{Ablowitz2004,Mio1976,Kodama1985}. Considering higher order nonlinear effects, three celebrated equations with derivative–type
nonlinearities have been proposed, namely, the Kaup-Newell (KN) equation \cite{KN1978} , the Chen-Lee-Liu (CLL) equation \cite{CLL1979} 
and the GI equation \cite{GI1983,GerdjikovKulish1983}. It is known that these three equations may be transformed into each other by a chain of gauge
transformations. Among these three equations, the GI equation takes the form
\begin{align}
\label{GI-eq}
\mathrm{i}u_{t}+u_{xx}-\mathrm{i}u^2u^*_{x}+\frac{1}{2}|u|^4u=0,
\end{align}
where $\mathrm{i}$ is the imaginary unit, asterisk denotes complex conjugation and $|\cdot|$ means the complex modulus. 
This equation has been studied extensively in terms of its Darboux transformation \cite{Fan2000}, Riemann–Hilbert approach \cite{GuoLing2012}, 
$\bar{\partial}$-dressing method \cite{LuoFan2020}, algebro-geometric solutions \cite{HeQiu2012} and Hamiltonian structures \cite{Fan2001}. 
The integrability of the GI equation is guaranteed by its Lax pair \cite{ZhuChen2021}
\begin{subequations}
\label{GI-Lax}
\begin{align}
\label{GI-Lax-x}
&\mathcal{X}_{x}=\mathcal{M}\mathcal{X}, \quad \mathcal{M}=\begin{pmatrix}
-\mathrm{i}\lambda^{2}+\frac{\mathrm{i}}{2}|u|^2 & \lambda u \\
-\lambda u^{*} & \mathrm{i}\lambda^{2}-\frac{\mathrm{i}}{2}|u|^2
\end{pmatrix},\\
&\mathcal{X}_{t}=\mathcal{N}\mathcal{X}, \quad \mathcal{N}=\begin{pmatrix}
N_{11} & N_{12} \\
N_{21} & -N_{11}
\end{pmatrix},
\end{align}
\end{subequations}
where $\lambda$ is the spectral parameter and
\begin{align*}
&N_{11}=-2\mathrm{i}\lambda^{4}+\mathrm{i}|u|^2\lambda^{2}-\frac{1}{2}(u^{*}u_{x}-uu^{*}_{x})+\frac{\mathrm{i}}{4}|u|^4,\\
&N_{12}=2\lambda^{3}u+\mathrm{i}\lambda u_{x},\\
&N_{21}=-2\lambda^{3}u^{*}+\mathrm{i}\lambda u^{*}_{x}.
\end{align*}
As the third-order flow in the GI hierarchy, the hGI equation reads \cite{Fan2001}
\begin{align}\label{hGI-eq}
u_{t}=-\frac{1}{2}u_{xxx}+\frac{3}{2}\mathrm{i}u|u_{x}|^2-\frac{3}{4}|u|^4u_{x},
\end{align}
which involves third-order dispersion and quintic nonlinearity. The Lax pair for the hGI equation shares the same spatial part \eqref{GI-Lax-x} but 
has a temporal part of the form $\mathcal{X}_{t}=\mathcal{N}\mathcal{X}$ with $\mathcal{N}$ being a polynomial matrix with $\lambda$ of
degree six \cite{ZhuChen2021}. In \cite{Zuo-Guo}, Zou et al. were concerned with the Riemann--Hilbert approach for 
the hGI equation \eqref{hGI-eq} with nonzero boundary conditions. Subsequently,
the soliton matrices of the hGI equation \eqref{hGI-eq} with simple zeros and elementary high-order zeros of Riemann--Hilbert 
problem were constructed through the standard dressing process. Besides, 
higher-order dispersion term affecting the propagation velocity, propagation direction and
amplitude of the soliton were revealed \cite{ZhuChen2021}. Recently, Shen et al. used the unified transformation approach to 
investigate the initial-boundary value problems of the model on the half-line \cite{Shen}. 

In recent decades, discrete integrable systems have emerged as a central topic in the theory of integrability \cite{AL1976,AL1977,Hirota1977}. 
These systems arise naturally in numerical analysis, combinatorics, and theoretical physics, and they often reveal 
structural features that are hidden in their continuous counterparts \cite{HJN2016}. 
Discrete analogues of classical soliton equations preserve the essential hallmarks of integrability, including multi-soliton solutions, 
Lax pairs, and B\"{a}cklund transformations, while living on lattice domains rather than continuous spacetime 
\cite{DJM,NQC,NCWQ1984,NC1995,ABS,ABS-2003,Atk,HZ-2009,But,Hie-BSQ,ZZN,NSZ}. 
Recently, a great deal of interest has been directed to the creation and study of discrete models of the NLS type equations, including 
nonlinear superposition principle of auto-B\"{a}cklund transformation \cite{Kono}, the discretized bilinear identity \cite{DJM-JPSJ}, 
symmetry-constraint method \cite{WH}, soliton solutions \cite{Fis}, a dNLS type Nijhoff--Quispel--Capel equation \cite{QNCL}, a dNLS equation and its modified version 
arising from the Cauchy matrix approach \cite{ZFJ-CNSNS,ZZS-TMPH}, and so on.
Compared with the numerous achievements in the discretization of NLS-type equations,
little work has been done on the discrete derivative–type NLS equations. Through reduction of the two-component KP hierarchy, 
a discrete CLL system is constructed, and its bilinear form is obtained accordingly \cite{DJM-JPSJ}. 
In \cite{Xu-KN}, Xu et al. constructed algebro-geometric solutions for this system. Besides, Xu et al. 
also propose a discrete KN system from the compatibility of two Darboux transformations of the KN spectral
problem \cite{Xu-CLL}. Recently, the Mikhailov model (negative order KN system) was discretized by means of the Cauchy matrix approach \cite{Mik}, in which 
one equation in the system comes from the compatibility of the two
Miura transformations and the other is transformed from the discrete negative order Ablowitz--Kaup--Newell--Segur (AKNS) system by using the Miura transformations.  

The Cauchy matrix method provides a powerful and systematic tool for constructing integrable equations and their solutions simultaneously. 
This approach was introduced for soliton solutions of the Adler-Bobenko-Suris list \cite{ABS} by Nijhoff et al. \cite{NAH2009} and later adapted to 
continuous, semi-discrete and discrete settings \cite{XZZ,ZhangZhao2013,Zhao2016,ZS-ZNA,Zhao2018,Mes,HLZ2026,SLZZ2025}. Its key idea is to start from the Sylvester equation
\[\bA\bX-\bX\bB=\bC,\]
together with appropriately designed dispersion relations for the matrix $\bC$, and to derive 
closed-form integrable equations through master functions defined via the resolvent of the Sylvester equation. 
The same framework automatically yields multi-soliton solutions when the spectral matrices are chosen to be diagonal,
and multiple-pole (or higher-order pole) solutions when Jordan-block spectral matrices are employed. 
The Cauchy matrix approach has been successfully applied to lattice versions of the Korteweg-de Vries, Boussinesq, and AKNS-type 
equations \cite{ZhangZhao2013,Zhao2016,ZZS-CMB,ZS-ZNA}, 
demonstrating its versatility and algorithmic nature. More recently, the method has also been extended to construct nonlocal integrable 
equations and their solutions \cite{ZhangZhao2022nonlocal,XZ2022,HSZ2025}.

In this paper, we systematically construct discrete versions of both the GI equation and the hGI equation within the 
Cauchy matrix framework. Starting from the Sylvester equation equipped with block matrix structures, we introduce two distinct sets of 
discrete dispersion relations, denoted by SDE-I and SDE-II, which produce different shift dynamics for the master functions. By specializing 
these shift relations to low-order master functions and eliminating auxiliary variables through algebraic identities, we arrive at closed 
lattice systems. A key observation is that the elimination step is not unique: the expressions for certain shift differences admit multiple 
equally valid representations, each leading to a different but conjugate-symmetric discrete system. This procedure yields four families of 
dGI models from SDE-I and four families of dhGI models from SDE-II. This multiplicity reveals an 
intrinsic non-uniqueness in the discretization of derivative-type integrable equations, a phenomenon that deserves further investigation.

For each of the discrete models obtained, we provide explicit $N$-soliton solutions and $N$-th 
order multiple-pole solutions via the Cauchy matrix method. The soliton solutions correspond to diagonal 
spectral matrices, while the multiple-pole solutions arise from Jordan-block choices. We further verify the 
consistency of our discrete models by performing a two-step continuum limit, contracting one lattice direction at a time. 
In the first step, the fully discrete model reduces to a semi-discrete GI (or semi-discrete hGI (sdhGI)) model; in the second step, the continuous GI (or hGI) system 
is recovered. This limiting procedure confirms that all four dGI models reduce to the same continuous GI equation and all four dhGI models reduce 
to the same continuous hGI equation.

Finally, we investigate symmetry reductions. All four dGI models and dhGI models are conjugate-symmetric,
which allows one to impose local complex conjugate reductions that produce scalar dGI and dhGI equations with explicit solutions.
Moreover, in the higher-order case, we discover that suitable pairwise recombinations of the dhGI lattice equations admit nonlocal reductions, 
leading to nonlocal dhGI equations. This feature is absent in the standard dGI case and highlights a structural distinction of the higher-order hierarchy.

The remainder of this paper is organized as follows. In Section \ref{2}, we introduce the Cauchy matrix scheme based on the Sylvester equation, 
define the master functions, and present their fundamental properties including recurrence relations and similarity invariance. In Section \ref{3},
we construct the four dGI models from SDE-I, derive their soliton and multiple-pole solutions, establish their continuum limits 
leading to the GI equation, and discuss local conjugate reductions. In Section \ref{4}, we carry out the analogous construction for the 
dhGI models from SDE-II, including solutions, continuum limits, and both local and nonlocal reductions. Concluding remarks are given 
in Section \ref{5}. There are four appendices. The first one presents the KP-type Cauchy matrix scheme, which offers an alternative derivation for the dhGI models, 
while the remaining three appendices contain supplementary expressions for the scalar dGI equations, scalar dhGI equations, and nonlocal dhGI equations.

\section{Cauchy matrix scheme}
\label{2}

The Cauchy matrix scheme starts from a Sylvester equation together with
properly defined plane wave vectors and a set of master functions.
In what follows, we show how this approach works in the discretization of GI type equations.

We start from the following Sylvester equation
\begin{equation}
\label{SE}
\bK\bM-\bM\bK=\br\tc,
\end{equation}
where $\bK\in \mathbb{C}^{2 N\times 2N}$ and  $\tc \in \mathbb{C}^{2\times 2N}$ are constant matrices, 
$\bM\in \mathbb{C}^{2N\times 2N}[n,m]$,
$\br \in \mathbb{C}^{2N\times 2}[n,m]$ are matrix functions of $(n,m)\in \mathbb{Z}^2$. These matrices are all block matrices with the following form 
\begin{align}\label{KMrc}
\bK=\begin{pmatrix}
\bK_{1}&\\
&\bK_{2}
\end{pmatrix},\quad\bM=\begin{pmatrix}
&\bM_{1}\\
\bM_{2}&
\end{pmatrix},\quad\br=\begin{pmatrix}
\br_{1}&\\
&\br_{2}
\end{pmatrix},\quad\tc=\begin{pmatrix}
&\tc_{2}\\
\tc_{1}&
\end{pmatrix},
\end{align}
in which for $i=1,2$, $\bK_{i}$ and $\bM_{i}$ are $N\times N$ matrices, $\br_{i}$ are $N$-dimensional column vectors, $\tc_{i}$ are $N$-dimensional row vectors.
 Based on the above framework, the Sylvester equation \eqref{SE} can be decoupled as
\begin{align}
\label{SE12}
\bK_{1}\bM_{1}-\bM_{1}\bK_{2}=\br_{1}\tc_{2},\quad\bK_{2}\bM_{2}-\bM_{2}\bK_{1}=\br_{2}\tc_{1}
\end{align}
by employing algebraic block matrix operations. Thus, $\bM_{1}$ and $\bM_{2}$ can be uniquely determined respectively if we assume 
$\bK_{1}$ and $\bK_{2}$ do not share eigenvalues \cite{Syl}. 
In the light of the Sylvester equation \eqref{SE}, we state the following proposition. 
\begin{prop}
For the matrix $\bM$ satisfying the Sylvester equation \eqref{SE}, the following recurrence relations hold
\begin{subequations}
\begin{align}
\label{ksm}
&\bK^{s}\bM-\bM\bK^{s}=\sum_{l=0}^{s-1}\bK^{s-1-l}\br\tc\bK^{l},\quad s\in\mathbb{Z}^{+},\\
\label{-ksm}
&\bK^{-s}\bM-\bM\bK^{-s}=-\sum_{l=-1}^{-s}\bK^{-s-1-l}\br\tc\bK^{l},\quad s\in\mathbb{Z}^{+},
\end{align}
\end{subequations}
where $\bK^0$ denotes the $2N\times 2N$ identity matrix.
\end{prop}

For the sake of brevity, we appoint that $\bI$ is the identity matrix, whose index indicating the size is omitted.
Introduce a set of functions
\begin{equation}
\label{Sij}
\bS^{(i,j)}=\tc\bK^{j}(\bI+\bM)^{-1}\bK^{i}\br=\begin{pmatrix}
s_{1}&s_{2}\\
s_{3}&s_{4}
\end{pmatrix}, \quad i,j\in\mathbb{Z},
\end{equation}
where we require that the matrix $\bI + \bM$ be formally invertible.
Note that each $\bS^{(i,j)}$ is a $2\times 2$ matrix.
We call $\{\bS^{(i,j)}\}$ master functions because they are the main objects
used to characterize integrable equations in the Cauchy matrix approach. With the structures of $\bM,\bK,\br$ and $\tc$ in \eqref{KMrc} specified, 
it follows that $s_{i}~(i=1,2,3,4)$ admit explicit expressions
\begin{subequations}
\label{si}
\begin{align}
&s_{1}=-\tc_{2}\bK_{2}^{j}\bM_{2}(\bI-\bM_{1}\bM_{2})^{-1}\bK_{1}^{i}\br_{1},\\
&s_{2}=\tc_{2}\bK_{2}^{j}(\bI-\bM_{2}\bM_{1})^{-1}\bK_{2}^{i}\br_{2},\\
&s_{3}=\tc_{1}\bK_{1}^{j}(\bI-\bM_{1}\bM_{2})^{-1}\bK_{1}^{i}\br_{1},\\
&s_{4}=-\tc_{1}\bK_{1}^{j}\bM_{1}(\bI-\bM_{2}\bM_{1})^{-1}\bK_{2}^{i}\br_{2}.
\end{align}
\end{subequations}

In what follows, we list two important properties of the master
functions $\{\bS^{(i,j)}\}$: recurrence relations and similarity invariance (see \cite{Zhao2016}). 
In the context of dGI and dhGI models, these are essential for the derivation and solution processes, respectively.
\begin{prop}
For the master functions $\{\bS^{(i,j)}\}$ defined in \eqref{Sij} with $\bM, \bK,\br$ and $\tc$
satisfying the Sylvester equation \eqref{SE},
the following recurrence relations hold
\begin{subequations}
\label{Sij-re}
\begin{align}
\label{Sij-re-po}
\bS^{(i,j+s)}&=\bS^{(i+s,j)}-\sum_{l=0}^{s-1}\bS^{(s-l-1,j)}\bS^{(i,l)},\quad s\in\mathbb{Z}^{+},\\
\label{Sij-re-ne}
\bS^{(i,j-s)}&=\bS^{(i-s,j)}+\sum_{l=-1}^{-s}\bS^{(-s-l-1,j)}\bS^{(i,l)},\quad s\in\mathbb{Z}^{+}.
\end{align}
\end{subequations}
In particular, when $s=1$ we have
\begin{subequations}
\label{Sij-re=1}
\begin{align}
\label{Sij-re=1-po}
\bS^{(i,j+1)}=&\bS^{(i+1,j)}-\bS^{(0,j)}\bS^{(i,0)},\\
\label{Sij-re=1-ne}
\bS^{(i,j-1)}=&\bS^{(i-1,j)}+\bS^{(-1,j)}\bS^{(i,-1)}.
\end{align}
\end{subequations}
\end{prop}

These recurrence relations will be used in our derivation.

\begin{prop}\label{prop-2-3}
It suffices to consider only the canonical forms of $\bK$, since any matrix
similar to $\bK$ leads to the same $\bS^{(i,j)}$.
In fact, let
\begin{align}
\label{trans-KLM}
\bar{\bK}=\bT\bK\bT^{-1},\quad \bar{\bM}=\bT\bM\bT^{-1},\quad\bar{\br}=\bT\br,\quad\bar{\tc}=\tc\bT^{-1}, 
\end{align}
where $\bT$ is the transform matrix. One can verify that
\begin{equation}
\label{SE-bar}
\bar{\bK}\bar{\bM}-\bar{\bM}\bar{\bK}=\bar{\br}\bar{\tc}
\end{equation}
and
\begin{equation}
\bS^{(i,j)}=\tc\bK^{j}(\bI+\bM)^{-1}\bK^{i}\br
=\bar{\tc}\bar{\bK}^{j}(\bI+\bar{\bM})^{-1}\bar{\bK}^{i}\bar{\br},
\end{equation}
which show that $\{\bS^{(i,j)}\}$ are similarity invariant for $\bK$ and its any similar matrix.
\end{prop}

\begin{remark}\label{Re-2-1}
In the following investigation, we always assume that the matrices $\bK_1$ and $\bK_2$ in the 
Sylvester equation \eqref{SE12} do not share any common eigenvalues, ensuring the unique determinability 
of $\bM_{1}$ and $\bM_{2}$. By introducing different dispersion relations for $\br$, we will then proceed to 
construct the dGI and dhGI models within the Cauchy matrix framework.
\end{remark}

\section{Discrete GI models}\label{3}

In this section, we construct discrete analogues of the GI system within the Cauchy matrix framework. 
Starting from the Sylvester equation together with the first set of discrete dispersion relations, 
we first derive the shift dynamics of the master functions. These shift relations are then used to eliminate 
auxiliary quantities and to obtain four conjugate-symmetric dGI models. We further present their Cauchy matrix solutions, 
including soliton and multiple-pole solutions, and examine their continuum limits and local conjugate reductions. 
The methodology and results developed here will serve as the foundation for the parallel construction in Section \ref{4}, 
where a second set of dispersion relations leads to discrete analogues of the hGI system.

\subsection{Shift dynamics of master function}

In the discrete context, functions are defined on discrete independent variables $n,m,\cdots$,
and derivatives are replaced by shifts with respect to these independent variables.
Before we proceed, let us recall some conventional notations used in discrete equations.
Suppose $f(n,m)$ is a function (can be a matrix function) defined on $\mathbb{Z}^2$. We denote
\begin{equation}
f:=f(n,m),~~ \wt{f}:=f(n+1,m),~~ \wh{f}:=f(n,m+1),~~ \wh{\wt{f}}:=f(n+1,m+1).
\end{equation}

To construct the dGI models, we consider the Sylvester equation together with the discrete dispersion relations given below
\begin{subequations}
\label{SDE-I}
\begin{align}
\label{SE1}
&\bK\bM-\bM\bK=\br\tc,\\
\label{t-r}
& \bP(\bP^*-\bK)\wt{\br}=\bP^*(\bP-\bK)\br,\\
\label{h-r}
&\bQ(\bQ^*-\bK)\wh{\br}=\bQ^*(\bQ-\bK)\br,
\end{align}
\end{subequations}
where
\begin{align}
\label{bPbQ}
\bP=\begin{pmatrix}
p\bI&\boldsymbol{0}\\\boldsymbol{0}&p^*\bI
\end{pmatrix},\quad \bQ=\begin{pmatrix}
q\bI&\boldsymbol{0}\\\boldsymbol{0}&q^*\bI
\end{pmatrix},
\end{align}
and $\bK,\bM,\br$ and $\tc$ are block matrices defined in \eqref{KMrc}. We refer to \eqref{SDE-I} as the first set of 
determining equations (SDE-I). To prepare for the proof of the proposition \ref{prop-3-1} below, we first establish the identities that connect $(\bP,\bQ)$ with $(\bM,\br,\tc)$:
\begin{align}
\label{commu}
\bP\bM=\bM\bP^*,\quad \bQ\bM=\bM\bQ^*, \quad \bP\br=\br\bp, \quad \bQ\br=\br\bq,\quad \tc\bP=\bp^*\tc, \quad \tc\bQ=\bq^*\tc,
\end{align}
where
\begin{subequations}
\begin{align}\label{bpbq}
\bp=\begin{pmatrix}
p&0\\0&p^*
\end{pmatrix},\quad\bq=\begin{pmatrix}
q&0\\0&q^*
\end{pmatrix}.
\end{align}
\end{subequations}
Based on the SDE-I \eqref{SDE-I}, we obtain the following propositions.
\begin{prop}\label{prop-3-1}
For $\bK,\bM,\br$ and $\tc$ satisfying the SDE-I \eqref{SDE-I}, one has
\begin{subequations}
\begin{align}
\label{t-M}
& \bP(\bP^*-\bK)\wt{\bM}=\bP^*(\bP-\bK)\bM,\\
\label{h-M}
&\bQ(\bQ^*-\bK)\wh{\bM}=\bQ^*(\bQ-\bK)\bM,
\end{align}
\end{subequations}
\begin{subequations}
and
\begin{align}
\label{t-M-1}
&\wt{\bM}\bP^*(\bP-\bK)-\bP^*(\bP-\bK)\bM=\wt{\br}\bp\tc,\\
\label{t-M-2}
&\bM\bP(\bP^*-\bK)-\bP(\bP^*-\bK)\wt{\bM}=\br\bp^*\tc,\\
\label{h-M-1}
&\wh{\bM}\bQ^*(\bQ-\bK)-\bQ^*(\bQ-\bK)\bM=\wh{\br}\bq\tc,\\
\label{h-M-2}
&\bM\bQ(\bQ^*-\bK)-\bQ(\bQ^*-\bK)\wh{\bM}=\br\bq^*\tc.
\end{align}
\end{subequations}
\end{prop}
\begin{proof}
Taking $\wt{\phantom{a}}$-shift of \eqref{SE1} and incorporating equation \eqref{t-r}, we obtain
\begin{align}
\bK\bP(\bP^*-\bK)\wt{\bM}-\bP(\bP^*-\bK)\wt{\bM}\bK=\bP^*(\bP-\bK)\br\tc.
\end{align}
Substituting \eqref{SE1} for $\br\tc$ in the above equation yields
\begin{align}
\bK\big[ \bP(\bP^*-\bK)\wt{\bM}-\bP^*(\bP-\bK)\bM\big] -\big[\bP(\bP^*-\bK)\wt{\bM}-\bP^*(\bP-\bK)\bM\big]\bK=\bm 0,
\end{align}
which indicates that \eqref{t-M} holds under the remark \ref{Re-2-1}. Equation \eqref{h-M} is obtained in a similar manner.

We now proceed to prove the relations \eqref{t-M-1}–\eqref{h-M-2}. 
Left-multiplying $\wt{\eqref{SE1}}$ by the matrix $\bP$ and identifying \eqref{commu}, we have
\begin{align}
\bP\bK\wt{\bM}-\bP\wt{\bM}\bK=\wt{\br}\bp\tc.
\end{align}
Replacing $\bP\bK\wt{\bM}$ by \eqref{t-M} gives
\begin{align}
\wt{\br}\bp\tc=&\bP\bP^*\wt{\bM}-\bP^*(\bP-\bK)\bM-\bP\wt{\bM}\bK \nn \\
=&\wt{\bM}\bP^*(\bP-\bK)-\bP^*(\bP-\bK)\bM.
\end{align}
We can also directly left-multiply \eqref{SE1} by $\bP^*$ and then replace $\bP^*\bK\bM$ with \eqref{t-M} to get
\begin{align}
\br\bp^*\tc=&\bP\bP^*\bM-\bP(\bP^*-\bK)\wt{\bM}-\bP^*\bM\bK \nn \\
=&\bM\bP(\bP^*-\bK)-\bP(\bP^*-\bK)\wt{\bM}.
\end{align}
The same proof method applies to \eqref{h-M-1} and \eqref{h-M-2}.
\end{proof}

\begin{prop}\label{prop-3-2}
Suppose that matrices $\bK, \bM, \br,$ and $\tc$ satisfy the SDE-I \eqref{SDE-I}, 
and thereby the following shift relations for $\bS^{(i,j)}$ as defined in \eqref{Sij} hold 
\begin{subequations}
\label{sh-Sij}
\begin{align}
\label{p-Sij-1}
&|p|^2\wt{\bS}^{(i,j)}-\bp^*\wt{\bS}^{(i,j+1)}
=|p|^2\bS^{(i,j)}-\bS^{(i+1,j)}\bp^*+\bS^{(0,j)}\bp^*\wt{\bS}^{(i,0)},\\
\label{p-Sij-2}	
&|p|^2\bS^{(i,j)}-\bp\bS^{(i,j+1)}=|p|^2\wt{\bS}^{(i,j)}-\wt{\bS}^{(i+1,j)}\bp+\wt{\bS}^{(0,j)}\bp\bS^{(i,0)},\\
\label{q-Sij-1}
&|q|^2\wh{\bS}^{(i,j)}-\bq^*\wh{\bS}^{(i,j+1)}
=|q|^2\bS^{(i,j)}-\bS^{(i+1,j)}\bq^*+\bS^{(0,j)}\bq^*\wh{\bS}^{(i,0)},\\
\label{q-Sij-2}
&|q|^2\bS^{(i,j)}-\bq\bS^{(i,j+1)}=|q|^2\wh{\bS}^{(i,j)}-\wh{\bS}^{(i+1,j)}\bq+\wh{\bS}^{(0,j)}\bq\bS^{(i,0)}.
\end{align}
\end{subequations}
\end{prop}
This proposition first appeared in \cite{ZZS-TMPH}.
The relations in Proposition \ref{prop-3-2} describe how the master functions evolve under the two lattice shifts. 
They provide the algebraic basis for the construction of the dGI models. In the next subsection, we specialize 
these relations to several low-order master functions and eliminate the auxiliary variables so 
that closed systems for the essential dependent variables can be obtained.

\subsection{Construction of the dGI models}\label{3-2}
To construct closed dGI systems, we focus on the low-order master functions $\bS^{(0,0)}$, $\bS^{(-1,0)}$ and $\bS^{(0,-1)}$. 
These functions contain the dependent variables from which the dGI models will be extracted. For convenience, we introduce
\begin{align}
\label{uvw-def}
&\bu=\bS^{(0,0)}=\begin{pmatrix}
u_{1}&u_{2}\\
u_{3}&u_{4}
\end{pmatrix}, \quad \bv=\bI-\bS^{(-1,0)}=\begin{pmatrix}
v_{1}&v_{2}\\
v_{3}&v_{4}
\end{pmatrix},\quad \bw=\bI+\bS^{(0,-1)},
\end{align}
where $u_i, v_{i}~(i=1,2,3,4)$ are scalar functions.
Next we work on the variables $\bu, \bv$ and $\bw$.

We observe that equation \eqref{Sij-re=1-ne} with specific
index combinations $(i,j)=(0,0)$,  $(1,0)$ and $(0,1)$ yields
the following relations in terms of $\bu, \bv$ and $\bw$:
\begin{align}
\label{vw}
\bv\bw=\bI, \quad\bS^{(1,-1)}=\bw\bu,\quad \bS^{(-1,1)}=\bu\bv. 
\end{align}
Among them, the first relation shows that $\bw$ is the inverse of $\bv$.
In addition, setting $(i, j)=(0, 0)$ in system \eqref{sh-Sij} yields
\begin{subequations}
\begin{align}
&|p|^2\wt{\bu}-\bp^*\wt{\bS}^{(0,1)}=|p|^2\bu-\bS^{(1,0)}\bp^*+\bu\bp^*\wt{\bu},\\
&|p|^2\bu-\bp\bS^{(0,1)}=|p|^2\wt{\bu}-\wt{\bS}^{(1,0)}\bp+\wt{\bu}\bp\bu,\\
&|q|^2\wh{\bu}-\bq^*\wh{\bS}^{(0,1)}=|q|^2\bu-\bS^{(1,0)}\bq^*+\bu\bq^*\wh{\bu},\\
&|q|^2\bu-\bq\bS^{(0,1)}=|q|^2\wh{\bu}-\wh{\bS}^{(1,0)}\bq+\wh{\bu}\bq\bu.
\end{align}
\end{subequations}
We eliminate $\bS^{(0,1)}$ and $\bS^{(1,0)}$ from the above relations to derive a matrix equation that contains only $\bu$:
\begin{align}\label{bu-close}
& \bp(|p|^2(\bu-\wt{\bu})\bq^*+|q|^2(\wh{\bu}-\bu)\bp^*+\bu(\bp^*\wt{\bu}\bq^*-\bq^*\wh{\bu}\bp^*))\bq \nn \\	
& \qquad\qquad\qquad\quad +\bq^*(|p|^2(\wh{\wt{\bu}}-\wh{\bu})\bq+|q|^2(\wt{\bu}-\wh{\wt{\bu}})\bp+\wh{\wt{\bu}}(\bp\wh{\bu}\bq-\bq\wt{\bu}\bp))\bp=\bm 0.
\end{align}
Returning to the system \eqref{sh-Sij} and taking $(i, j) = (-1, 0)$, we obtain
\begin{subequations}
\label{th-uvw}
\begin{align}
\label{t-uvw-1}
&|p|^2\bv\wt{\bw}=|p|^2\bI+\bp^*\wt{\bu}-\bu\bp^*,\\
\label{t-uvw-2}
&|p|^2\wt{\bv}\bw=|p|^2\bI+\bp\bu-\wt{\bu}\bp,\\
\label{h-uvw-1}
&|q|^2\bv\wh{\bw}=|q|^2\bI+\bq^*\wh{\bu}-\bu\bq^*,\\
\label{h-uvw-2}
&|q|^2\wh{\bv}\bw=|q|^2\bI+\bq\bu-\wh{\bu}\bq,
\end{align}
\end{subequations}
in light of \eqref{vw}, which give rise to
\begin{subequations}
\begin{align}
\label{t-u}
&(|p|^2\bI+\bp^*\wt{\bu}-\bu\bp^*)(|p|^2\bI+\bp\bu-\wt{\bu}\bp)=|p|^4\bI,\\
\label{h-u}
&(|q|^2\bI+\bq^*\wh{\bu}-\bu\bq^*)(|q|^2\bI+\bq\bu-\wh{\bu}\bq)=|q|^4\bI.
\end{align}
\end{subequations}
These exhibit the evolution of $\bu$ along $n$- and $m$-directions, respectively.

The derivations thus far have been performed at the matrix level. Our objective, however, is to obtain scalar relations. 
To this end, we decompose certain matrix relations derived above by equating their corresponding components. 
Prior to this, we require explicit expressions for the four components of the matrix $\bw$. 
Invoking the definition of $\bv$ in \eqref{uvw-def} and using the Weinstein-Aronszajn formula (see Appendix D in \cite{HJN2016}) and 
the Sylvester equation \eqref{SE1}, one can find
\begin{align}
\label{det-v}
\det(\bv) &=\left| \bI-\bS^{(-1,0)}\right| \nn\\
&=\left|\bI-\tc(\bI+\bM)^{-1}\bK^{-1}\br\right| \nn \\
&= \left|\bI-(\bI+\bM)^{-1}\bK^{-1}\br\tc\right| \nn \\
&=\left|(\bI+\bM)^{-1}\bK^{-1}(\bI+\bM)\bK\right| \nn \\
&=1,
\end{align}
in light of \eqref{vw}. Moreover, $\bw$ defined in \eqref{uvw-def} can be expressed as
\begin{align}
\label{vw-def1}
\bw=\bI+\bS^{(0,-1)}=\begin{pmatrix}
v_{4}&-v_{2}\\
-v_{3}&v_{1}
\end{pmatrix}.
\end{align}
As a result, equations \eqref{h-uvw-1} and \eqref{h-uvw-2} can be represented in terms of $\{u_i\}$ and $\{v_i\}$:
\begin{subequations}
\label{h-uivi}
\begin{align}
\label{h-uivi-1}
&a_{1}+q^*(1+\wh{v}_{2}v_{3}-\wh{v}_{1}v_{4})=0,\\
\label{h-uivi-2}
&a_{1}-q(1+v_{2}\wh{v}_{3}-v_{1}\wh{v}_{4})=0,\\
\label{h-uivi-3}
&a_{2}-|q|^2(v_{1}\wh{v}_{2}-\wh{v}_{1}v_{2})=0,\\
\label{h-uivi-4}
&a_{3}-|q|^2(v_{3}\wh{v}_{4}-\wh{v}_{3}v_{4})=0,
\end{align}
where
\begin{align}
a_{1}:=u_{1}-\wh{u}_{1}=\wh{u}_{4}-u_{4},\quad a_{2}:=qu_{2}-q^*\wh{u}_2,\quad a_{3}:=q\wh{u}_{3}-q^*u_{3}.
\end{align}
\end{subequations}
In the next step, we introduce a new variable
\begin{equation}\label{nu}
\nu\coloneqq v_{2}/v_{4}.
\end{equation}
Then we revisit \eqref{h-uivi} by substituting  $v_{2}=v_{4}\nu$.
At first, for \eqref{h-uivi-3} we have
\begin{align}
\label{h-v5-1}
a_{2}=|q|^2(v_{1}\wh{v}_{4}\wh{\nu}-\wh{v}_{1}v_{4}\nu).
\end{align}
Combining \eqref{h-uivi-1} and \eqref{h-uivi-2} and eliminating $a_{1}$ yields
\begin{align}
q^*(1+\wh{v}_{2}v_{3}-\wh{v}_{1}v_{4})+q(1+v_{2}\wh{v}_{3}-v_{1}\wh{v}_{4})=0.
\end{align}
Then, using the above equation to replace $v_{1}\wh{v}_{4}$ in \eqref{h-v5-1} leads to
\begin{align}\label{h-v5-2}
a_{2}=|q|^2\big(\wh{\nu}(1+\theta_{q})+\nu\wh{\nu}v_{4}\wh{v}_{3}+\theta_{q}\wh{\nu}^2v_{3}\wh{v}_{4}-v_{4}\wh{v}_{1}(\theta_{q}\wh{\nu}+\nu)\big), 
\quad \theta_{q}=q^*/q.
\end{align}
Furthermore, we use \eqref{h-uivi-4} to replace $v_{3}\wh{v}_{4}$ and have
\begin{align}
\label{h-v5-3}
a_{2}=|q|^2\big(\wh{\nu}(1+\theta_{q})+(\wh{v}_{3}v_{4}\wh{\nu}-\wh{v}_{1}v_{4})(\theta_{q}\wh{\nu}+\nu)+a_{3}\wh{\nu}^2/q^{2}\big).
\end{align}
Again, replacing the term $\wh{v}_{3}$ by using \eqref{h-uivi-2} and performing some algebraic manipulations, 
we arrive at
\begin{align}
\label{a2}
a_{2}=|q|^2(\wh{\nu}-\nu)-a_{1}(q^*\wh{\nu}+q\nu)-a_{3}\nu\wh{\nu}.
\end{align}
On the other hand, expanding  \eqref{h-u} gives rise to the same relation in terms of $\{a_{i}\}$:
\begin{align}
\label{ai}
|q|^2(q-a_{1})(q^*+a_{1})+a_{2}a_{3}=|q|^4.
\end{align}
Thus we can insert \eqref{a2} into \eqref{ai} and then solve the equation for solution $a_{1}$, 
which is expressed as
\begin{align}
\label{a1}
2a_{1}=q-q^*-(\wh{u}_{3}-\theta_{q}u_{3})(\theta_{q}^{-1}\nu+\wh{\nu})+Q,
\end{align}
where
\begin{align}
\label{Q}
Q:=-\mathrm{i}[4|q|^2-(q+q^*+(\wh{u}_{3}-\theta_{q}u_{3})(\wh{\nu}-\theta_{q}^{-1}\nu))^2]^{\frac{1}{2}}.
\end{align}
Note that all relations that
we have derived for the  $\wh{\phantom{a}}$-shift (involving the lattice parameter $q$ and the lattice variable $m$)
hold also for the $\wt{\phantom{a}}$-shift, simply by replacing $q$ by $p$ and $\wh{\phantom{a}}$-shift by $\wt{\phantom{a}}$-shift. We therefore proceed to 
derive the relations for $\wt{\phantom{a}}$-shift directly
\begin{subequations}
\begin{align}
\label{b2}
&b_{2}=|p|^2(\wt{\nu}-\nu)-b_{1}(p^*\wt{\nu}+p\nu)-b_{3}\nu\wt{\nu},\\
\label{bi}
&|p|^2(p-b_{1})(p^*+b_{1})+b_{2}b_{3}=|p|^4,\\
\label{b1}
&2b_{1}=p-p^*-(\wt{u}_{3}-\theta_{p}u_{3})(\theta_{p}^{-1}\nu+\wt{\nu})+P,\quad\theta_{p}=p^*/p,
\end{align}
\end{subequations}
where
\begin{subequations}
\begin{align}
&b_{1}:=u_{1}-\wt{u}_{1}=\wt{u}_{4}-u_{4},\quad b_{2}:=pu_{2}-p^*\wt{u}_2,\quad b_{3}:=p\wt{u}_{3}-p^*u_{3},\\
\label{P}
& P:=-\mathrm{i}[4|p|^2-(p+p^*+(\wt{u}_{3}-\theta_{p}u_{3})(\wt{\nu}-\theta_{p}^{-1}\nu))^2]^{\frac{1}{2}}.
\end{align}
\end{subequations}

\begin{remark}
A careful explanation is in order here. Previously, we defined $Q$ as given in \eqref{Q}, where the sign preceding the imaginary unit $\mathrm{i}$ is negative. In fact, the choice of this sign depends on the sign of the imaginary part of the complex number $q$. A detailed summary is as follows. 
\begin{itemize}
\item When $\operatorname{Im}(q)>0$, $Q$ is defined as
\begin{align}
\label{Q-1}
Q:=-\mathrm{i}[4|q|^2-(q+q^*+(\wh{u}_{3}-\theta_{q}u_{3})(\wh{
		\nu}-\theta_{q}^{-1}\nu))^2]^{\frac{1}{2}};
\end{align}
\item 
When $\operatorname{Im}(q)<0$, it is defined as
\begin{align}
\label{Q-2}
Q:=\mathrm{i}[4|q|^2-(q+q^*+(\wh{u}_{3}-\theta_{q}u_{3})(\wh{
		\nu}-\theta_{q}^{-1}\nu))^2]^{\frac{1}{2}}.
\end{align}
\end{itemize}
The same convention applies to $P$. According to the above statement, the expressions for $Q$ and $P$ in \eqref{Q} and \eqref{P} imply that $\operatorname{Im}(q)>0$ and $\operatorname{Im}(p)>0$.

\end{remark}

Using the relations derived above, we construct a discrete model that involves only $u_{3}$ and $\nu$, 
this constitutes a discrete version of the GI model. In fact, one relation involving only $u_{3}$ and $\nu$ can be derived 
from the matrix equation \eqref{bu-close} by taking its $(2, 1)$-component. To derive a second, 
independent equation of these two variables, we taking \eqref{a2} and \eqref{b2} into the following identity 
\begin{align}
pa_{2}-qb_{2}=p^*\wt{a}_{2}-q^*\wh{a}_{2}.
\end{align}
As a result, the following two equations can be obtained
\begin{subequations}
\label{u3-v-u1}
\begin{align}
& |p|^2(q(\wh{u}_{3}-\wh{\wt{u}}_{3})+q^*(\wt{u}_{3}-u_{3}))+|q|^2(p(\wh{\wt{u}}_{3}-\wt{u}_{3})+p^*(u_{3}-\wh{u}_{3})) \nn \\
& \qquad +(p^*q\wh{u}_{3}-pq^*\wt{u}_{3})(\wh{\wt{u}}_{1}-u_{1})+(p^*q^*u_{3}-pq\wh{\wt{u}}_{3})(\wh{u}_{1}-\wt{u}_{1})=0, \\
& |p|^2(q^*(\wh{\nu}-\wh{\wt{\nu}})+q(\wt{\nu}-\nu))+|q|^2(p^*(\wh{\wt{\nu}}-\wt{\nu})+p(\nu-\wh{\nu}))+(pq^*\wh{\nu}-p^*q\wt{\nu})\nn \\
& \qquad \times (u_{1}-\wh{\wt{u}}_{1}+\wh{\wt{\nu}}\wh{\wt{u}}_{3}-\nu u_{3})
+(pq\nu-p^*q^*\wh{\wt{\nu}})(\wt{u}_{1}-\wh{u}_{1}+\wh{\nu}\wh{u}_{3}-\wt{\nu}\wt{u}_{3})=0.
\end{align}
\end{subequations}
Both equations, however, involve terms containing $u_{1}$; therefore, to proceed, it is necessary to eliminate these terms by introducing relations
that depend exclusively on $u_{3}$ and $\nu$. We find that
\begin{subequations}
\label{u1-1}
\begin{align}
\label{u1-1a}
& u_{1}-\wh{\wt{u}}_{1}=a_{1}+\wh{b}_{1}, \\
\label{u1-1b}
& \wt{u}_{1}-\wh{u}_{1}=a_{1}-b_{1}.
\end{align} 
\end{subequations}

By replacing \eqref{a1} and \eqref{b1} with the identity above, finally we construct a lattice system with dependent variables $(u_3,\nu):=(u_{3},v_{2}/v_{4})$:
\begin{subequations}
\label{dGI}
\begin{align}
\label{dGI1}
&\begin{aligned}
&2|p|^2(q(\wh{u}_{3}-\wh{\wt{u}}_{3})+q^*(\wt{u}_{3}-u_{3}))+2|q|^2(p(\wh{\wt{u}}_{3}-\wt{u}_{3})+p^*(u_{3}-\wh{u}_{3}))+(p^*q\wh{u}_{3}-pq^*\wt{u}_{3})\times\\
&(q^*-q+p^*-p+\theta_{q}^{-1}\wh{u}_{3}\nu-\theta_{q}u_{3}\wh{\nu}+\theta_{p}^{-1}\wh{\wt{u}}_{3}\wh{\nu}-\theta_{p}\wh{u}_{3}\wh{\wt{\nu}}+\wh{\wt{u}}_{3}\wh{\wt{\nu}}-u_{3}\nu-Q-\wh{P})+(p^*q^*u_{3}-pq\wh{\wt{u}}_{3})\\
&\times(q^*-q+p-p^*-\theta_{p}^{-1}\wt{u}_{3}\nu+\theta_{p}u_{3}\wt{\nu}+\theta_{q}^{-1}\wh{u}_{3}\nu-\theta_{q}u_{3}\wh{\nu}+\wh{u}_{3}\wh{\nu}-\wt{u}_{3}\wt{\nu}-Q+P)=0,
\end{aligned}\\
\label{dGI2}
&\begin{aligned}
&2|p|^2(q^*(\wh{\nu}-\wh{\wt{\nu}})+q(\wt{\nu}-\nu))+2|q|^2(p^*(\wh{\wt{\nu}}-\wt{\nu})+p(\nu-\wh{\nu}))+(pq^*\wh{\nu}-p^*q\wt{\nu})\times\\
&(q-q^*+p-p^*+\theta_{q}u_{3}\wh{\nu}-\theta_{q}^{-1}\wh{u}_{3}\nu+\theta_{p}\wh{u}_{3}\wh{\wt{\nu}}-\theta_{p}^{-1}\wh{\wt{u}}_{3}\wh{\nu}+\wh{\wt{u}}_{3}\wh{\wt{\nu}}-u_{3}\nu+Q+\wh{P})+(pq\nu-p^*q^*\wh{\wt{\nu}})\\
&\times(q-q^*+p^*-p-\theta_{p}u_{3}\wt{\nu}+\theta_{p}^{-1}\wt{u}_{3}\nu+\theta_{q}u_{3}\wh{\nu}-\theta_{q}^{-1}\wh{u}_{3}\nu+\wh{u}_{3}\wh{\nu}-\wt{u}_{3}\wt{\nu}+Q-P)=0,
\end{aligned}
\end{align}
\end{subequations}
where $P$ and $Q$ are defined in \eqref{P} and \eqref{Q}. The above system may be interpreted as a dGI model, 
since its continuous limit yields the GI equation (see Subsection 3.4). Observing that this dGI model is conjugate symmetric, 
it is of interest to consider its local reduction (see Subsection 3.5).
\begin{remark}
In \eqref{u1-1}, the expressions for $(u_{1}-\wh{\wt{u}}_{1})$ and $(\wt{u}_{1}-\wh{u}_{1})$ in terms of $a_{1}$ and $b_{1}$ are not unique. 
Beyond the relation given in \eqref{u1-1}, we also have the following representation of $(u_{1}-\wh{\wt{u}}_{1})$ and $(\wt{u}_{1}-\wh{u}_{1})$:
\begin{subequations} 
\label{u1-2}
\begin{align}
\label{u1-2a}
& u_{1}-\wh{\wt{u}}_{1}=\wt{a}_{1}+b_{1}, \\
\label{u1-2b}
& \wt{u}_{1}-\wh{u}_{1}=\wt{a}_{1}-\wh{b}_{1}.
\end{align}
\end{subequations}
Consequently, substituting $(u_{1}-\wh{\wt{u}}_{1})$ and $(\wt{u}_{1}-\wh{u}_{1})$ in \eqref{u3-v-u1} leaves us with four feasible possibilities, 
from which we derive four sets of dGI models. We elaborate them in the following.

\noindent$\bullet$~~For (\eqref{u1-1a}, \eqref{u1-1b}): the corresponding dGI model is presented in \eqref{dGI}.

\noindent $\bullet$~~For (\eqref{u1-1a}, \eqref{u1-2b}):
\begin{subequations}
\label{dGI-2}
\begin{align}
\label{dGI1-2}
&\begin{aligned}
&2|p|^2(q(\wh{u}_{3}-\wh{\wt{u}}_{3})+q^*(\wt{u}_{3}-u_{3}))+2|q|^2(p(\wh{\wt{u}}_{3}-\wt{u}_{3})+p^*(u_{3}-\wh{u}_{3}))+(p^*q\wh{u}_{3}-pq^*\wt{u}_{3})\times\\
&(q^*-q+p^*-p+\theta_{q}^{-1}\wh{u}_{3}\nu-\theta_{q}u_{3}\wh{\nu}+\theta_{p}^{-1}\wh{\wt{u}}_{3}\wh{\nu}-\theta_{p}\wh{u}_{3}\wh{\wt{\nu}}+\wh{\wt{u}}_{3}\wh{\wt{\nu}}-u_{3}\nu-Q-\wh{P})+(p^*q^*u_{3}-pq\wh{\wt{u}}_{3})\\
&\times(q^*-q+p-p^*-\theta_{p}^{-1}\wh{\wt{u}}_{3}\wh{\nu}+\theta_{p}\wh{u}_{3}\wh{\wt{\nu}}+\theta_{q}^{-1}\wh{\wt{u}}_{3}\wt{\nu}-\theta_{q}\wt{u}_{3}\wh{\wt{\nu}}+\wh{u}_{3}\wh{\nu}-\wt{u}_{3}\wt{\nu}-\wt{Q}+\wh{P})=0,
\end{aligned}\\
\label{dGI2-2}
&\begin{aligned}
&2|p|^2(q^*(\wh{\nu}-\wh{\wt{\nu}})+q(\wt{\nu}-\nu))+2|q|^2(p^*(\wh{\wt{\nu}}-\wt{\nu})+p(\nu-\wh{\nu}))+(pq^*\wh{\nu}-p^*q\wt{\nu})\times\\
&(q-q^*+p-p^*+\theta_{q}u_{3}\wh{\nu}-\theta_{q}^{-1}\wh{u}_{3}\nu+\theta_{p}\wh{u}_{3}\wh{\wt{\nu}}-\theta_{p}^{-1}\wh{\wt{u}}_{3}\wh{\nu}+\wh{\wt{u}}_{3}\wh{\wt{\nu}}-u_{3}\nu+Q+\wh{P})+(pq\nu-p^*q^*\wh{\wt{\nu}})\\
&\times(q-q^*+p^*-p+\theta_{p}^{-1}\wh{\wt{u}}_{3}\wh{\nu}-\theta_{p}\wh{u}_{3}\wh{\wt{\nu}}-\theta_{q}^{-1}\wh{\wt{u}}_{3}\wt{\nu}+\theta_{q}\wt{u}_{3}\wh{\wt{\nu}}+\wh{u}_{3}\wh{\nu}-\wt{u}_{3}\wt{\nu}+\wt{Q}-\wh{P})=0.
\end{aligned}
\end{align}
\end{subequations}

\noindent $\bullet$~~For (\eqref{u1-2a}, \eqref{u1-1b}):
\begin{subequations}
\label{dGI-3}
\begin{align}
\label{dGI1-3}
&\begin{aligned}
&2|p|^2(q(\wh{u}_{3}-\wh{\wt{u}}_{3})+q^*(\wt{u}_{3}-u_{3}))+2|q|^2(p(\wh{\wt{u}}_{3}-\wt{u}_{3})+p^*(u_{3}-\wh{u}_{3}))+(p^*q\wh{u}_{3}-pq^*\wt{u}_{3})\times\\
&(q^*-q+p^*-p+\theta_{q}^{-1}\wh{\wt{u}}_{3}\wt{\nu}-\theta_{q}\wt{u}_{3}\wh{\wt{\nu}}+\theta_{p}^{-1}\wt{u}_{3}\nu-\theta_{p}u_{3}\wt{\nu}+\wh{\wt{u}}_{3}\wh{\wt{\nu}}-u_{3}\nu-\wt{Q}-P)+(p^*q^*u_{3}-pq\wh{\wt{u}}_{3})\\
&\times(q^*-q+p-p^*-\theta_{p}^{-1}\wt{u}_{3}\nu+\theta_{p}u_{3}\wt{\nu}+\theta_{q}^{-1}\wh{u}_{3}\nu-\theta_{q}u_{3}\wh{\nu}+\wh{u}_{3}\wh{\nu}-\wt{u}_{3}\wt{\nu}-Q+P)=0,
\end{aligned}\\
\label{dGI2-3}
&\begin{aligned}
&2|p|^2(q^*(\wh{\nu}-\wh{\wt{\nu}})+q(\wt{\nu}-\nu))+2|q|^2(p^*(\wh{\wt{\nu}}-\wt{\nu})+p(\nu-\wh{\nu}))+(pq^*\wh{\nu}-p^*q\wt{\nu})\times\\
&(q-q^*+p-p^*-\theta_{q}^{-1}\wh{\wt{u}}_{3}\wt{\nu}+\theta_{q}\wt{u}_{3}\wh{\wt{\nu}}-\theta_{p}^{-1}\wt{u}_{3}\nu+\theta_{p}u_{3}\wt{\nu}+\wh{\wt{u}}_{3}\wh{\wt{\nu}}-u_{3}\nu+\wt{Q}+P)+(pq\nu-p^*q^*\wh{\wt{\nu}})\\
&\times(q-q^*+p^*-p-\theta_{p}u_{3}\wt{\nu}+\theta_{p}^{-1}\wt{u}_{3}\nu+\theta_{q}u_{3}\wh{\nu}-\theta_{q}^{-1}\wh{u}_{3}\nu+\wh{u}_{3}\wh{\nu}-\wt{u}_{3}\wt{\nu}+Q-P)=0.
\end{aligned}
\end{align}
\end{subequations}

\noindent $\bullet$~~For (\eqref{u1-2a}, \eqref{u1-2b}):
\begin{subequations}
\label{dGI-4}
\begin{align}
\label{dGI1-4}
&\begin{aligned}
&2|p|^2(q(\wh{u}_{3}-\wh{\wt{u}}_{3})+q^*(\wt{u}_{3}-u_{3}))+2|q|^2(p(\wh{\wt{u}}_{3}-\wt{u}_{3})+p^*(u_{3}-\wh{u}_{3}))+(p^*q\wh{u}_{3}-pq^*\wt{u}_{3})\times\\
&(q^*-q+p^*-p+\theta_{q}^{-1}\wh{\wt{u}}_{3}\wt{\nu}-\theta_{q}\wt{u}_{3}\wh{\wt{\nu}}+\theta_{p}^{-1}\wt{u}_{3}\nu-\theta_{p}u_{3}\wt{\nu}+\wh{\wt{u}}_{3}\wh{\wt{\nu}}-u_{3}\nu-\wt{Q}-P)+(p^*q^*u_{3}-pq\wh{\wt{u}}_{3})\\
&\times(q^*-q+p-p^*-\theta_{p}^{-1}\wh{\wt{u}}_{3}\wh{\nu}+\theta_{p}\wh{u}_{3}\wh{\wt{\nu}}+\theta_{q}^{-1}\wh{\wt{u}}_{3}\wt{\nu}-\theta_{q}\wt{u}_{3}\wh{\wt{\nu}}+\wh{u}_{3}\wh{\nu}-\wt{u}_{3}\wt{\nu}-\wt{Q}+\wh{P})=0,
\end{aligned}\\
\label{dGI2-4}
&\begin{aligned}
&2|p|^2(q^*(\wh{\nu}-\wh{\wt{\nu}})+q(\wt{\nu}-\nu))+2|q|^2(p^*(\wh{\wt{\nu}}-\wt{\nu})+p(\nu-\wh{\nu}))+(pq^*\wh{\nu}-p^*q\wt{\nu})\times\\
&(q-q^*+p-p^*-\theta_{q}^{-1}\wh{\wt{u}}_{3}\wt{\nu}+\theta_{q}\wt{u}_{3}\wh{\wt{\nu}}-\theta_{p}^{-1}\wt{u}_{3}\nu+\theta_{p}u_{3}\wt{\nu}+\wh{\wt{u}}_{3}\wh{\wt{\nu}}-u_{3}\nu+\wt{Q}+P)+(pq\nu-p^*q^*\wh{\wt{\nu}})\\
&\times(q-q^*+p^*-p+\theta_{p}^{-1}\wh{\wt{u}}_{3}\wh{\nu}-\theta_{p}\wh{u}_{3}\wh{\wt{\nu}}-\theta_{q}^{-1}\wh{\wt{u}}_{3}\wt{\nu}+\theta_{q}\wt{u}_{3}\wh{\wt{\nu}}+\wh{u}_{3}\wh{\nu}-\wt{u}_{3}\wt{\nu}+\wt{Q}-\wh{P})=0.
\end{aligned}
\end{align}
\end{subequations}
The above three models, together with \eqref{dGI}, are all conjugate symmetric, with similar but not identical forms. 
We regard all four models as distinct dGI models, since each is derived from the unified discrete equation \eqref{u3-v-u1}, whose continuous limit yields the GI system.
\end{remark}

\begin{remark}
Alternatively, we may use a pair of new dependent variables $(v_{3}/v_{1},u_{2})$ to construct a discrete GI model identical to 
\eqref{dGI} (or \eqref{dGI-2}, \eqref{dGI-3}, \eqref{dGI-4}). This indicates that $(v_{3}/v_{1},u_{2})$ are also solutions of the
discrete GI models \eqref{dGI}, \eqref{dGI-2}, \eqref{dGI-3}, \eqref{dGI-4}, where $v_{3}/v_{1}$ corresponds to 
$u_{3}$ and $u_{2}$ to $\nu$ in the original models. In the following study, we present only the solution expressed in terms of the 
dependent variables $(u_{3},\nu)$ without loss of generality.
\end{remark}	

\subsection{Solutions to the dGI models}
\label{3-3}

Having derived the dGI models, we now turn to their explicit solutions. The Cauchy matrix framework not only yields the equations but also provides their 
solutions once the SDE-I \eqref{SDE-I} is solved. In particular, diagonal form of the spectral matrices $\bK_{1}$ and $\bK_{2}$ 
lead to soliton solutions, whereas Jordan-block form give multiple-pole solutions. At this stage we introduce the discrete plane wave factors
\begin{subequations}
\label{plane-def}
\begin{align}
&\rho(z)=\left(\theta_{p}\frac{p-z}{p^*-z}\right)^{n}
\left(\theta_{q}\frac{q-z}{q^*-z}\right) ^{m}\rho^{0}(z),\\
&\sigma(z)=\left(\theta_{p}\frac{p-z}{p^*-z}\right)^{-n}
\left(\theta_{q}\frac{q-z}{q^*-z}\right)^{-m}\sigma^{0}(z),
\end{align}
\end{subequations}
where $\rho^{0}(z),~\sigma^{0}(z)$ are phase parameters independent of $(n,m)$. 

Returning to SDE-I \eqref{SDE-I} and in light of the setting \eqref{KMrc}, the SDE-I decouples into the following form
\begin{subequations}
\label{SDE-I-1}
\begin{align}
\label{SE-1}
&\bK_{1}\bM_{1}-\bM_{1}\bK_{2}=\br_{1}\tc_{2},\quad \bK_{2}\bM_{2}-\bM_{2}\bK_{1}=\br_{2}\tc_{1}, \\
&p\bI(p^*\bI-\bK_{1})\wt{\br}_{1}=p^*\bI(p\bI-\bK_{1})\br_{1},\quad p^*\bI(p\bI-\bK_{2})\wt{\br}_{2}=p\bI(p^*\bI-\bK_{2})\br_{2},\\
&q\bI(q^*\bI-\bK_{1})\wh{\br}_{1}=q^*\bI(q\bI-\bK_{1})\br_{1},\quad q^*\bI(q\bI-\bK_{2})\wh{\br}_{2}=q\bI(q^*\bI-\bK_{2})\br_{2},
\end{align}
\end{subequations}
where it is assumed that $\bK_1$ and $\bK_2$ do not share any eigenvalues. The matrix functions $\bu$ and $\bv$ defined in equation \eqref{uvw-def} are formulated as
\begin{align}
\label{sym-uv-def}
\bu=\bS^{(0,0)}=\tc(\bI+\bM)^{-1}\br,\quad\bv=\bI-\bS^{(-1,0)}
=\bI-\tc(\bI+\bM)^{-1}\bK^{-1}\br.
\end{align}
For the solution $(u_{3},\nu)$ to the  dGI model \eqref{dGI},
their formulae are given by
\begin{align}
\label{uv1}
(u_{3},\nu)=\left(u_{3},v_{2}/v_{4}\right) =\left( \tc_{1}(\bI-\bM_{1}\bM_{2})^{-1}\br_{1},
\frac{-\tc_{2}(\bI-\bM_{2}\bM_{1})^{-1}\bK_{2}^{-1}\br_{2}}{1+\tc_{1}
\bM_{1}(\bI-\bM_{2}\bM_{1})^{-1}\bK_{2}^{-1}\br_{2}}\right),
\end{align}
where the involved elements are determined by the system \eqref{SDE-I-1}.

Next, we present the explicit formulae for soliton and multiple-pole solutions
by solving \eqref{SDE-I-1}. These are summarized in the following theorems.
\begin{theorem}[\textbf{Soliton solutions}]
Functions $(u_{3},\nu)$ defined in \eqref{uv1} serve as $N$-soliton solutions for the dGI model \eqref{dGI} with
\begin{subequations}
\begin{align}
&\bK_{1}=\mathrm{diag}(k_{1},k_{2},\cdots,k_{N}), \quad \bK_{2}=\mathrm{diag}(l_{1},l_{2},\cdots,l_{N}),\\
&\br_{1}=(\rho(k_{1}),\rho(k_{2}),\cdots,\rho(k_{N}))^{\st}, \quad \br_{2}=(\sigma(l_{1}),\sigma(l_{2}),\cdots,\sigma(l_{N}))^{\st},\\
&\tc_{1}=(c_{1,1},c_{1,2},\cdots,c_{1,N}), \quad \tc_{2}=(c_{2,1},c_{2,2},\cdots,c_{2,N}),\\
&\bM_{1}=(M_{1,ij})_{N\times N},\quad M_{1,ij}=\frac{\rho(k_{i})c_{2,j}}{k_{i}-l_{j}},\\
&\bM_{2}=(M_{2,ij})_{N\times N},\quad M_{2,ij}=\frac{\sigma(l_{i})c_{1,j}}{l_{i}-k_{j}},
\end{align}
\end{subequations}
where the discrete plane wave factors $\rho(z)$ and $\sigma(z)$ are defined in \eqref{plane-def}.
\end{theorem}
\begin{theorem}[\textbf{Multiple-pole solutions}]
Functions $(u_{3},\nu)$ defined in \eqref{uv1} serve as $N$-th order multiple-pole solutions for the dGI model \eqref{dGI} with
\begin{subequations}
\begin{align}
&\bK_{1}=\bJ[k,1],\quad\bK_{2}=\bJ[l,1],\quad \mathrm{with}\quad k\neq l,\\
&\br_{1}=\bF_{k}[\rho(k)]\be_{N},\quad\br_{2}=\bF_{l}[\sigma(l)]\be_{N},\quad \quad \be_{N}=(\underbrace{1, 0, \cdots, 0}_{N\text{-dimensional}})^{\st},\\
&\tc_{1}=\be^{\st}_{N}\bH[c_{1,j}],\quad\tc_{2}=\be^{\st}_{N}\bH[c_{2,j}],\\
&\bM_{1}=\bF_{k}[\rho(k)]\bG_{1}\bH[c_{2,j}],\quad\bM_{2}=\bF_{l}[\sigma(l)]\bG_{2}\bH[c_{1,j}],\\
\label{G1G2}
&\bG_{1}=\left( \frac{(-1)^{i-1}
\mathrm{C}^{i-1}_{i+j-2}}{(k-l)^{i+j-1}}\right) _{N\times N},\quad \bG_{2}=\left( \frac{(-1)^{i-1}
\mathrm{C}^{i-1}_{i+j-2}}{(l-k)^{i+j-1}}\right) _{N\times N},
\end{align}
\end{subequations}
where $C_{n}^m$ denotes the binomial coefficient and
\begin{align}
\label{J-def}
&\bJ[z,\delta]:=\begin{pmatrix}
z&0&0&\cdots&0\\
\delta&z&0&\cdots&0\\
0&\delta&z&\cdots&0\\
\vdots&\vdots&\vdots&\ddots&\vdots\\
0&0&0&\cdots&z
\end{pmatrix}_{N\times N}, \quad  \delta=\pm 1, \\
\label{F-def}
&\bF_{\delta z}[f(z)]=\begin{pmatrix}
f(z)&0&0&\cdots&0\\
\partial_{\delta z}f(z)&f(z)&0&\cdots&0\\
\frac{\partial^{2}_{\delta z}f(z)}{2!}&\partial_{\delta z}f(z)&f(z)
&\cdots&0\\
\vdots&\vdots&\vdots&\ddots&\vdots\\
\frac{\partial^{N-1}_{\delta z}f(z)}{(N-1)!}
&\frac{\partial^{N-2}_{\delta z}f(z)}{(N-2)!}
&\frac{\partial^{N-3}_{\delta z}f(z)}{(N-3)!}&\cdots&f(z)
\end{pmatrix}_{N\times N}, \\
\label{H-def}
&\bH[b_{j}]=
\begin{pmatrix}
b_{1} & \cdots 
& b_{N-2} & b_{N-1} &b_{N}\\
b_{2}
& \cdots & b_{N-1} & b_{N} & 0\\
b_{3} & \cdots &b_{N}& 0 & 0\\
\vdots & \vdots & \vdots & \ddots &\ddots\\
b_{N} & \cdots & 0 & 0 & 0
\end{pmatrix}_{N\times N}.
\end{align}
\end{theorem}

\begin{remark}
$\bF_{\delta z}[f(z)]$ in \eqref{F-def} is a lower triangular Toeplitz matrix commuting with a Jordan block of same order.
All such matrices compose a commutative set $\c{G}^{[\mathcal{N}]}$ with respect to matrix multiplication
and the subset
\[G^{[\mathcal{N}]}=\big \{\mathcal{C} \big |~\big. \mathcal{C}\in \c{G}^{[\mathcal{N}]},~|\mathcal{C}|\neq 0 \big\}\]
is an Abelian group. Such kind of matrices play useful roles in the expression of multiple-pole solutions.
For more properties of such matrices one can refer to \cite{Z-arxiv-2006,ZZSZ-RMP-2014}.
$\bH[b_{j}]$ in \eqref{H-def} is a skew triangular Hankel matrix, which is symmetric. It is easy to find that
$\bH[b_{j}]\bF_{\delta z}[f(z)]$ is a symmetric matrix, i.e.
$\bH[b_{j}]\bF_{\delta z}[f(z)]=\bF^{\st}_{\delta z}[f(z)]\bH[b_{j}]$.
\end{remark}

\subsection{Continuum limits: from the dGI models to the GI system}
We next examine the continuum limits of the dGI models. By contracting one lattice direction, the fully discrete models reduce to a semi-discrete GI (sd-GI) system. A further contraction of the remaining lattice direction then yields the continuous GI system. This two-step limiting procedure also shows that the four dGI models obtained in Subsection \ref{3-2} share the same continuous GI system, despite their different lattice forms. We note that an analogous two-step procedure will be applied to the dhGI models in Section \ref{4}, where the different scaling of the dispersion relations leads to the higher-order GI system rather than the standard one.

For this purpose, we take $p=a+a\mathrm{i}$ and $q=b+b\mathrm{i}$ with $a,b>0$; consequently, the discrete plane wave factors $\rho(z)$ and $\sigma(z)$
from \eqref{plane-def} may be expressed as:
\begin{subequations}
\label{plane-def-1}
\begin{align}
&\rho(z)=\left(\frac{a-a\mathrm{i}+z\mathrm{i}}{a-a\mathrm{i}-z}\right)^{n}
\left(\frac{b-b\mathrm{i}+z\mathrm{i}}{b-b\mathrm{i}-z}\right) ^{m}\rho^{0}(z),\\
&\sigma(z)=\left(\frac{a-a\mathrm{i}-z}{a-a\mathrm{i}+z\mathrm{i}}\right)^{n}
\left(\frac{b-b\mathrm{i}-z}{b-b\mathrm{i}+z\mathrm{i}}\right)^{m}\sigma^{0}(z).
\end{align}
\end{subequations}

\subsubsection{The sd-GI model}
In the derivation of the dGI models, $p$ and $q$ are arbitrary complex numbers. Here, by setting $p=a+a\mathrm{i}$ and $q=b+b\mathrm{i}$, we present only the necessary relations that are directly relevant to the sd-GI model.

Upon setting $p=a+a\mathrm{i}$ and $q=b+b\mathrm{i}$, the discrete equations \eqref{u3-v-u1} admit a further simplification to
\begin{subequations}
\label{u3-v-u1-1}
\begin{align}
& b(1+\mathrm{i})(\mathrm{i}u_{3}-\mathrm{i}\wh{u}_{3}+\wt{u}_{3}-\wh{\wt{u}}_{3})-a(1+\mathrm{i})(\mathrm{i}u_{3}+\wh{u}_{3}-\mathrm{i}\wt{u}_{3}-\wh{\wt{u}}_{3}) \nn \\
& \qquad\qquad +(\wh{u}_{3}-\wt{u}_{3})(u_{1}-\wh{\wt{u}}_{1})+\mathrm{i}(u_{3}+\wh{\wt{u}}_{3})(\wh{u}_{1}-\wt{u}_{1})=0, \\
&  b(1+\mathrm{i})(\nu-\wh{\nu}+\mathrm{i}\wt{\nu}-\mathrm{i}\wh{\wt{\nu}})-a(1+\mathrm{i})(\nu+\mathrm{i}\wh{\nu}-\wt{\nu}-\mathrm{i}\wh{\wt{\nu}}) \nn \\
& \qquad\qquad  +(\wh{\nu}-\wt{\nu})(u_{1}-\wh{\wt{u}}_{1}-u_{3}\nu+\wh{\wt{u}}_{3}\wh{\wt{\nu}})+\mathrm{i}(\nu+\wh{\wt{\nu}})(\wt{u}_{1}-\wh{u}_{1}+\wh{u}_{3}\wh{\nu}-\wt{u}_{3}\wt{\nu})=0,
\end{align}
\end{subequations}
and equation \eqref{b1} reduces to
\begin{align}
\label{b1-1}
2b_{1}=2\mathrm{i}a+(u_{3}-\mathrm{i}\wt{u}_{3})(\nu-\mathrm{i}\wt{\nu})-\mathrm{i}\left[ 8a^2-(2a+(u_{3}-\mathrm{i}\wt{u}_{3})(\nu+\mathrm{i}\wt{\nu}))^2\right] ^{\frac{1}{2}}.
\end{align}
Next, our attention is directed to the continuum limits of \eqref{u3-v-u1-1} and \eqref{b1-1}.

We let both $b$ and $m$ tend to infinity but keep $m/b$ finite.
Introduce
\begin{equation}\label{xmq}
\xi:=m/b.
\end{equation}
At this point, $\xi$ is the new continuous coordinate.
With this setting we reinterpret the variables $u_1$, $u_{3}$ and $\nu$ as
\begin{align}
u_{1}(n,m)=u_{1}(n,\xi),\quad u_{3}(n,m)=u_{3}(n,\xi),\quad\nu(n,m)=\nu(n,\xi)
\end{align}
without making confusions,
while the shifted variables in $m$-direction give rise to
\begin{subequations}
\label{Taylor-expan-uv}
\begin{align}
\label{Taylor-expan-u}
\wh{u}_{3}=&u_{3}(n,\xi+1/b)=u_{3}+u_{3,\xi}/b+u_{3,\xi\xi}/(2b^{2})+\cdots,\\
\label{Taylor-expan-wu}
\wh{\wt{u}}_{3}=&u_{3}(n+1,\xi+1/b)=\wt{u}_{3}+\wt{u}_{3,\xi}/b+\wt{u}_{3,\xi\xi}/(2b^{2})+\cdots,
\end{align}
\end{subequations}
and a similar formula  for  $\wh{u}_{1},\wh{\wt{u}}_{1},\wh{\nu}$ and $\wh{\wt{\nu}}$.
Inserting them into the system \eqref{u3-v-u1-1} and then
from the leading term (in terms of $\mathcal{O}(b^0)$) we obtain a semi-discrete system:
\begin{subequations}
\label{sd-u3u1v}
\begin{align}
&a(1+\mathrm{i})(u_{3}-\wt{u}_{3})+\mathrm{i}u_{3,\xi}+\wt{u}_{3,\xi}-(u_{1}-\wt{u}_{1})(u_{3}+\mathrm{i}\wt{u}_{3})=0,\\
&a(1+\mathrm{i})(\nu-\wt{\nu})+\nu_{\xi}+\mathrm{i}\wt{\nu}_{\xi}+(u_{1}-\wt{u}_{1}+\wt{u}_{3}\wt{\nu}-u_{3}\nu)(\mathrm{i}\nu+\wt{\nu})=0.
\end{align}
\end{subequations}
Observing that the term $(u_{1}-\wt{u}_{1})$ in the semi-discrete system \eqref{sd-u3u1v} is in fact $b_{1}$, 
we replace it by means of equation \eqref{b1-1} and rearrange the resulting system into a semi-discrete form involving only $u_{3}$ and $\nu$:
\begin{subequations}
\label{sd-u3v}
\begin{align}
&2a(1+\mathrm{i})(u_{3}-\wt{u}_{3})+2\mathrm{i}u_{3,\xi}+2\wt{u}_{3,\xi}-(2a\mathrm{i}+(u_{3}-\mathrm{i}\wt{u}_{3})(\nu-\mathrm{i}\wt{\nu})-\mathrm{i}\mathcal{P})(u_{3}+\mathrm{i}\wt{u}_{3})=0,\\
&2a(1-\mathrm{i})(\nu-\wt{\nu})-2\mathrm{i}\nu_{\xi}+2\wt{\nu}_{\xi}+(2a\mathrm{i}-(u_{3}+\mathrm{i}\wt{u}_{3})(\nu+\mathrm{i}\wt{\nu})-\mathrm{i}\mathcal{P})(\nu-\mathrm{i}\wt{\nu})=0,
\end{align}
where
\begin{align}
\mathcal{P}:=\left[ 8a^2-(2a+(u_{3}-\mathrm{i}\wt{u}_{3})(\nu+\mathrm{i}\wt{\nu}))^2\right] ^{\frac{1}{2}}.
\end{align}
\end{subequations}
We refer to \eqref{sd-u3v} as the sd-GI model. In Section \ref{3-2}, four dGI models were derived. The continuum limits performed above show that all four dGI models reduce to the same sd-GI model \eqref{sd-u3v}. Below, we present the solutions of this sd-GI model and briefly describe the Cauchy matrix method adapted to it.

In the above limit, the discrete plane wave factors in \eqref{plane-def-1} give rise to
\begin{subequations}
\label{sd-plane}
\begin{align}
\label{sd-plane-1}
&\rho(z)=\left(\frac{a-a\mathrm{i}+z\mathrm{i}}{a-a\mathrm{i}-z}\right)^{n}e^{\mathrm{i}z\xi}\rho^{0}(z),\\
\label{sd-plane-2}
&\sigma(z)=\left(\frac{a-a\mathrm{i}-z}{a-a\mathrm{i}+z\mathrm{i}}\right)^{n}e^{-\mathrm{i}z\xi}\sigma^{0}(z),
\end{align}
\end{subequations}
which now serve for the sd-GI model \eqref{sd-u3v}.

\begin{prop}\label{Prop-4-1}
Formulae \eqref{uv1} provide solutions to
the sd-GI model \eqref{sd-u3v}, through the Cauchy matrix scheme:
\begin{subequations}
\begin{align}
&\bK_{1}\bM_{1}-\bM_{1}\bK_{2}=\br_{1}\tc_{2},\quad \bK_{2}\bM_{2}-\bM_{2}\bK_{1}=\br_{2}\tc_{1}, \\
&(a(1-\mathrm{i})\bI-\bK_{1})\wt{\br}_{1}=(a(1-\mathrm{i})\bI+\mathrm{i}\bK_{1})\br_{1},\quad\br_{1,\xi}=\mathrm{i}\bK_{1}\br_{1},\\
&(a(1-\mathrm{i})\bI+\mathrm{i}\bK_{2})\wt{\br}_{2}=(a(1-\mathrm{i})\bI-\bK_{2})\br_{2}, \quad\br_{2,\xi}=-\mathrm{i}\bK_{2}\br_{2}.
\end{align}
\end{subequations}
\end{prop}

\subsubsection{Continuous GI system}
Now, we apply the second continuum limit to the sd-GI model \eqref{sd-u3v} to obtain the continuous GI system. To this end, we set
\begin{align}
\label{lim-n}
n\rightarrow\infty,\quad a\rightarrow\infty,\quad\text{while}\quad \tau:=n/a  \quad\text{finite},
\end{align}
along with
\begin{align}
u_{1}(n,\xi)=u_{1}(\tau,\xi),\quad u_{3}(n,\xi)=u_{3}(\tau,\xi),\quad\nu(n,
\xi)=\nu(\tau,\xi),
\end{align}
where $\tau$ is another continuous coordinate. To obtain nontrivial equations, we introduce new coordinates defined in terms of their derivatives
\begin{align}
\label{xi-tau}	\partial_{\tau}=\frac{1}{2a}\partial_{t}+\partial_{x},\quad\partial_{\xi}=\partial_{x}.
\end{align}
We may expand
\begin{align}
\wt{u}_{3}=u_{3}+u_{3,\tau}/a+u_{3,\tau\tau}/(2a^2)+\cdots=u_{3}+u_{3,x}/a+(u_{3,t}+u_{3,xx})/(2a^2)+\cdots,
\end{align} 
and a similar formula  for $\wt{u}_{1}$ and $\wt{\nu}$. Inserting them into \eqref{b1-1} and \eqref{sd-u3u1v}, the leading order (in terms of 
$\mathcal{O}(a)$) yield a coupled system
\begin{subequations}
\begin{align}
&u_{1,x}-u_{3}\nu_{x}+\mathrm{i}u_{3}^2\nu^2=0,\\
&u_{3,t}+\mathrm{i}u_{3,xx}-2u_{3}u_{1,x}=0,\\
&\nu_{t}-\mathrm{i}\nu_{xx}-2\nu(u_{3}\nu_{x}+u_{3,x}\nu-u_{1,x})=0,
\end{align}
\end{subequations}
which imply
\begin{subequations}
\label{ctn-u3v}
\begin{align}
&u_{3,t}+\mathrm{i}u_{3,xx}-2u_{3}^2\nu_{x}+2\mathrm{i}u_{3}^3\nu^2=0,\\
&\nu_{t}-\mathrm{i}\nu_{xx}-2\nu^2u_{3,x}-2\mathrm{i}\nu^3u_{3}^2=0.
\end{align}
\end{subequations}
The system \eqref{ctn-u3v} is the well‑known GI system.

Solving the GI system \eqref{ctn-u3v} requires studying the variation of the plane wave factor \eqref{sd-plane} 
under the continuum limit \eqref{lim-n} and the coordinate transformation \eqref{xi-tau}. We must expand \eqref{sd-plane-1} as
\begin{align}
\begin{aligned}
\rho(z)=&\left(\frac{a-a\mathrm{i}+z\mathrm{i}}{a-a\mathrm{i}-z}\right)^{n}e^{\mathrm{i}z\xi}\rho^{0}(z)\\
=&\exp\left\lbrace \mathrm{i}z\xi+a\tau\left[ \ln\left(1+ \frac{z\mathrm{i}}{a-a\mathrm{i}}\right)-\ln\left(1- \frac{z}{a-a\mathrm{i}}\right)\right]  \right\rbrace\rho^{0}(z) \\
=&\exp\left\lbrace \mathrm{i}z(\xi+\tau)+\frac{\mathrm{i}z^2\tau}{2a}+\tau\mathcal{O}(a^{-2})\right\rbrace\rho^{0}(z) \\
=&\exp\left\lbrace \mathrm{i}zx+\mathrm{i}z^2t+\mathcal{O}(a^{-1})\right\rbrace\rho^{0}(z),
\end{aligned}
\end{align}
an analogous expansion form for \eqref{sd-plane-2} is given by
\begin{align}
\sigma(z)=\exp\left\lbrace -\mathrm{i}zx-\mathrm{i}z^2t+\mathcal{O}(a^{-1})\right\rbrace\sigma^{0}(z).
\end{align}
Consequently, solutions for the GI system \eqref{ctn-u3v} can also be obtained from
the Cauchy matrix scheme.
\begin{prop}
Formulae \eqref{uv1} provide solutions to
the continuous GI system \eqref{ctn-u3v}
through the Cauchy matrix scheme:
\begin{subequations}
\begin{align}
&\bK_{1}\bM_{1}-\bM_{1}\bK_{2}=\br_{1}\tc_{2},\quad \bK_{2}\bM_{2}-\bM_{2}\bK_{1}=\br_{2}\tc_{1}, \\
&\br_{1,x}=\mathrm{i}\bK_{1}\br_{1},\quad \br_{1,t}=\mathrm{i}\bK_{1}^2\br_{1},\\
&\br_{2,x}=-\mathrm{i}\bK_{2}\br_{2},\quad \br_{2,t}=-\mathrm{i}\bK_{2}^2\br_{2}.
\end{align}
\end{subequations}
\end{prop}

\subsection{Local conjugate reduction}

The four dGI models obtained in Subsection \ref{3-2} are complex conjugate symmetric. 
This symmetry allows one to impose the local reduction $\nu=u_{3}^*$, 
under which each two-component system reduces to a scalar dGI equation. In this subsection, 
we first write down the reduced equations and then derive the corresponding reduction constraints 
at the solution level. In this way, the solutions (including soliton and multipole solutions) 
of the scalar dGI equation are obtained. We remark that only local reductions are admissible here; 
the reason is that the complex conjugate lattice parameters in SDE-I break the parity symmetry required 
for nonlocal reductions. In contrast, the dhGI models constructed in Section \ref{4}, with their real lattice parameters, 
will be shown to support both local and nonlocal reductions.

We begin by presenting the scalar dGI equations. Applying the reduction $\nu=u_{3}^*$ to \eqref{dGI},  
\eqref{dGI-2}, \eqref{dGI-3} and \eqref{dGI-4}, we arrive at four scalar dGI equations. As a representative example, 
the reduction of the first model \eqref{dGI} yields
\begin{align}
\label{d-gi-1}
&2|p|^2(q(\wh{u}_{3}-\wh{\wt{u}}_{3})+q^*(\wt{u}_{3}-u_{3}))+2|q|^2(p(\wh{\wt{u}}_{3}-\wt{u}_{3})+p^*(u_{3}-\wh{u}_{3}))+(p^*q\wh{u}_{3}-pq^*\wt{u}_{3}) \nn \\
& \times(q^*-q+p^*-p+\theta_{q}^{-1}\wh{u}_{3}u_{3}^*-\theta_{q}u_{3}\wh{u}_{3}^*+\theta_{p}^{-1}\wh{\wt{u}}_{3}\wh{u}_{3}^*-\theta_{p}\wh{u}_{3}\wh{\wt{u}}_{3}^*
+|\wh{\wt{u}}_{3}|^2-|u_{3}|^2-Q_{1}-\wh{P}_{1}) \nn \\
&+(p^*q^*u_{3}-pq\wh{\wt{u}}_{3})
(q^*-q+p-p^*-\theta_{p}^{-1}\wt{u}_{3}u_{3}^*+\theta_{p}u_{3}\wt{u}_{3}^*+\theta_{q}^{-1}\wh{u}_{3}u_{3}^*-\theta_{q}u_{3}\wh{u}_{3}^* \nn \\
& +|\wh{u}_{3}|^2-|\wt{u}_{3}|^2-Q_{1}+P_{1})=0,
\end{align}
where 
\begin{subequations}
\label{PQ-1}
\begin{align}
P_{1}=-\mathrm{i}[4|p|^2-(p+p^*+|\wt{u}_{3}-\theta_{p}u_{3}|^2)^2]^{\frac{1}{2}}, \\
Q_{1}=-\mathrm{i}[4|q|^2-(q+q^*+|\wh{u}_{3}-\theta_{q}u_{3}|^2)^2]^{\frac{1}{2}}.
\end{align}
\end{subequations}

The remaining three scalar dGI equations, obtained from \eqref{dGI-2}--\eqref{dGI-4} under the same reduction, are listed in Appendix \ref{App-B}.
Although the four scalar equations are very similar, they differ from one another. Each constitutes a dGI equation, showing that the discretization of the GI equation 
is not unique. For the solutions derived from the Cauchy matrix method in Subsection \ref{3-3}, 
we present the reduction conditions and the corresponding solutions of the dGI equation in the theorem below.
\begin{theorem}\label{T-3-3}
The solution \eqref{uv1} admits the complex conjugate reduction $\nu=u_{3}^{*}$ under constraints
\begin{align}
\label{redu}
\bK_{2}=\bK_{1}^{*},\quad \br_{2}=-\bK_{2}\br_{1}^{*},\quad \tc_{2}=\tc_{1}^{*}.
\end{align}
\end{theorem}
\begin{proof}
By condition \eqref{redu}, from \eqref{SE-1} one immediately obtains
\begin{align}
\bK_{1}\bM_{1}-\bM_{1}\bK_{1}^{*}=\br_{1}\tc_{1}^{*},\quad
\bK_{1}(-\bK_{1}^{-1}\bM_{2}^{*})-(-\bK_{1}^{-1}\bM_{2}^{*})\bK_{1}^{*}=\br_{1}\tc_{1}^{*}.
\end{align}
Subtracting the second equation from the first one, we easily find the relations between $\bM_{1}$ and $\bM_{2}$: 
\begin{align}
\bM_{2}^{*}=-\bK_{1}\bM_{1},\quad\bM_{1}^{*}=-\bK_{2}^{-1}\bM_{2}\quad\Longrightarrow\quad\bM_{1}^{*}\bM_{2}^{*}=\bK_{2}^{-1}\bM_{2}\bK_{1}\bM_{1}.
\end{align}
Through straightforward calculations, it is demonstrated that
\begin{align*}
u_{3}^{*}v_{4}=&\tc_{1}^{*}(\bI-\bM_{1}^{*}\bM_{2}^{*})^{-1}\br_{1}^{*}\left[1+\tc_{1}\bM_{1}(\bI-\bM_{2}\bM_{1})^{-1}\bK_{2}^{-1}\br_{2}\right] \\
=&-\tc_{2}(\bI-\bM_{1}^{*}\bM_{2}^{*})^{-1}(\bI-\bM_{2}\bM_{1}+\bK_{2}^{-1}\br_{2}\tc_{1}\bM_{1})(\bI-\bM_{2}\bM_{1})^{-1}\bK_{2}^{-1}\br_{2}\\
=&-\tc_{2}(\bI-\bM_{1}^{*}\bM_{2}^{*})^{-1}(\bI-\bK_{2}^{-1}\bM_{2}\bK_{1}\bM_{1})(\bI-\bM_{2}\bM_{1})^{-1}\bK_{2}^{-1}\br_{2}\\
=&-\tc_{2}(\bI-\bM_{2}\bM_{1})^{-1}\bK_{2}^{-1}\br_{2}=v_{2}.
\end{align*}
Therefore, the proof is completed.
\end{proof}
\begin{theorem}
The function 
\begin{align}
u_{3}=\tc_{1}(\bI+\bM_{1}\bK^{*}_{1}\bM^{*}_{1})^{-1}\br_{1},
\end{align}
provides i)  $N$-soliton solutions for the dGI equation \eqref{d-gi-1} with
\begin{subequations}
\begin{align}
&\bK_{1}=\mathrm{diag}(k_{1},k_{2},\cdots,k_{N}),\quad k_{i}\in\mathbb{C}, \quad i=1,2,\cdots,N,\\
&\br_{1}=(\rho(k_{1}),\rho(k_{2}),\cdots,\rho(k_{N}))^{\st},\quad\tc_{1}=(c_{1,1},c_{1,2},\cdots,c_{1,N}),\\
&\bM_{1}=(M_{1,ij})_{N\times N},\quad M_{1,ij}=\frac{\rho(k_{i})c_{1,j}^*}{k_{i}-k_{j}^{*}};
\end{align}
\end{subequations}
and ii) the multiple-pole solutions for the dGI equation \eqref{d-gi-1} with
\begin{subequations}
\begin{align}
&\bK_{1}=\bJ[k,1],\quad\br_{1}=\bF_{k}[\rho(k)]\be_{N},\quad\tc_{1}=\be_{N}^{\st}\bH[c_{1,j}],\\
&\bM_{1}=\bF_{k}[\rho(k)]\bG\bH^{*}[c_{1,j}], \quad
\bG=(G_{i,j})_{N\times N},\quad G_{i,j}=\frac{(-1)^{i-1}C^{i-1}_{i+j-2}}{(k-k^{*})^{i+j-1}},
\end{align}
\end{subequations}
where the discrete plane wave factor $\rho(z)$ is defined in \eqref{plane-def},
and the matrices $\bJ,\bF,\bH$ are defined in \eqref{J-def}, \eqref{F-def} and \eqref{H-def}, respectively.
\end{theorem}

\begin{remark}
The reduction described above is a local one. Nonlocal reduction of the solutions cannot be achieved here, 
and we shall not pursue them further.
\end{remark}

The dGI models constructed above demonstrate that the Cauchy matrix framework, combined with SDE-I, 
successfully produces integrable discretizations of the GI system. The dispersion relations in SDE-I involve complex 
lattice parameters $(p, q)$ and their conjugates $(p^*, q^*)$, reflecting the NLS-type nature of the GI equation. 
A natural question arises: can the same Cauchy matrix philosophy, when equipped with a qualitatively different set of 
dispersion relations, generate discrete analogues of higher-order members in the GI hierarchy? The next section answers this affirmatively. 
By replacing SDE-I with a new set of determining equations (SDE-II), whose dispersion relations feature real lattice parameters and an additional
sign matrix $\ba$, we obtain dhGI models whose continuum limits yield the hGI system. 
The structural parallel between the two constructions underscores the universality of the Cauchy matrix approach, while the key 
differences in the dispersion relations, the form of the resulting lattice equations, and particularly the emergence of nonlocal 
reductions highlight the richer algebraic structure inherent in the higher-order members in hierarchy.

\section{Discrete higher-order GI models derived from SDE-II}\label{4}

This section extends the construction in Section \ref{3} to the hGI equation in hierarchy. 
Instead of using the SDE-I, we introduce a SDE-II. Although the derivation 
follows the same Cauchy matrix philosophy as before, the new determining equations lead to different shift dynamics and consequently to the
dhGI models. We derive four such models, present their soliton and multiple-pole solutions, and analyze both their continuum limits and reductions. 
In contrast to the dGI case, the dhGI models also admit nonlocal reductions after a suitable recombination of the lattice equations.

Compared with Section \ref{3}, the present section follows the same overall strategy but differs in three aspects. 
First, the determining equations are replaced by SDE-II, leading to different shift dynamics for the master functions. 
Second, the resulting continuum limit is not the standard GI system but its higher-order counterpart, the hGI system. 
Third, the dhGI models admit not only local conjugate reductions but also nonlocal reductions after suitable recombinations of the lattice equations.

\subsection{Shift dynamics of master function}

To derive the dhGI models, we replace SDE-I by the following SDE-II
\begin{subequations}
\label{SDE-II}
\begin{align}
&\bK\bM-\bM\bK=\br\tc, \\
& p\wt{\br}-\bK\wt{\br}\ba=p\br+\bK\br\ba,\\
& q\wh{\br}-\bK\wh{\br}\ba=q\br+\bK\br\ba,
\end{align}
\end{subequations}
where $p,q \in\mathbb{R}$, $\bK,\bM,\br$ and $\tc$ are defined as in \eqref{KMrc}, and the matrix $\ba=\mathrm{diag}(1,-1)$.
The essential difference from SDE-I lies in the new dispersion relations for $\br$. This modification changes the shift relations of the master functions. 
Starting from SDE-II, the shift relations for the master function $\bS^{(i,j)}$ defined in \eqref{Sij} are obtained and stated in the following proposition.
\begin{prop} \cite{ZS-ZNA}
Suppose that matrices $\bK, \bM, \br,$ and $\tc$ satisfy the SDE-II \eqref{SDE-II}, thereby the following shift relations for $\bS^{(i,j)}$ hold:
\begin{subequations}
\label{Sij1-ht}
\begin{align}
\label{Sij1-t1}
p\wt{\bS}^{(i,j)}+\ba\wt{\bS}^{(i,j+1)}=&p\bS^{(i,j)}+\bS^{(i+1,j)}\ba-\bS^{(0,j)}\ba\wt{\bS}^{(i,0)},\\
\label{Sij1-t2}
p\bS^{(i,j)}-\ba\bS^{(i,j+1)}=&p\wt{\bS}^{(i,j)}-\wt{\bS}^{(i+1,j)}\ba+\wt{\bS}^{(0,j)}\ba\bS^{(i,0)},\\
\label{Sij1-h1}
q\wh{\bS}^{(i,j)}+\ba\wh{\bS}^{(i,j+1)}=&q\bS^{(i,j)}+\bS^{(i+1,j)}\ba-\bS^{(0,j)}\ba\wh{\bS}^{(i,0)},\\
\label{Sij1-h2}
q\bS^{(i,j)}-\ba\bS^{(i,j+1)}=&q\wh{\bS}^{(i,j)}-\wh{\bS}^{(i+1,j)}\ba+\wh{\bS}^{(0,j)}\ba\bS^{(i,0)}.
\end{align}
\end{subequations}
\end{prop}
 Comparing the shift relations 
\eqref{Sij1-ht} with their SDE-I counterparts \eqref{sh-Sij}, the key structural difference becomes apparent: 
in SDE-I the shift relations involve conjugate lattice parameters $(\bp, \bp^*)$ and modular factors $|p|^2$, $|q|^2$, 
reflecting the complex nature of the dispersion, whereas in SDE-II the sign matrix $\ba$ takes their place and the lattice parameters $p, q$ remain real. 
This distinction propagates through the entire construction and ultimately determines whether the continuum limit yields the GI or the hGI system.

\subsection{Construction of dhGI models}

We now use the shift relations derived from SDE-II to construct closed lattice systems. 
The procedure parallels that of Subsection \ref{3-2}, but the resulting equations differ because the SDE-II dynamics 
involve the matrix $\ba$ and real lattice parameters $p,q$. We introduce the following variables 
\begin{align}
\label{UVW}
\bU=\bS^{(0,0)}=\begin{pmatrix}
U_{1}&U_{2}\\
U_{3}&U_{4}
\end{pmatrix}, \quad \bV=\bI-\bS^{(-1,0)}=\begin{pmatrix}
V_{1}&V_{2}\\
V_{3}&V_{4}
\end{pmatrix},\quad \bW=\bI+\bS^{(0,-1)},
\end{align}
where $U_i, V_{i}~(i=1,2,3,4)$ are scalar functions. It is evident that $\bU,\bV,\bW$ defined in this way still satisfy
\begin{align}
\label{VW}
\bV\bW=\bI, \quad\bS^{(1,-1)}=\bW\bU,\quad \bS^{(-1,1)}=\bU\bV,\quad \det(\bV)=1. 
\end{align}
These relations do not involve the $\wt{\phantom{a}}$-shift and the $\wh{\phantom{a}}$-shift, and thus bear a strong resemblance to those in \eqref{vw}. Nevertheless, one should distinguish them from $\bu,\bv,\bw$, respectively.

Proceeding further, we substitute $(i,j)=(0,0)$ into \eqref{Sij1-ht} and obtain
\begin{subequations}
\begin{align}
&p\wt{\bU}+\ba\wt{\bS}^{(0,1)}=p\bU+\bS^{(1,0)}\ba-\bU\ba\wt{\bU},\\
&p\bU-\ba\bS^{(0,1)}=p\wt{\bU}-\wt{\bS}^{(1,0)}\ba+\wt{\bU}\ba\bU,\\
&q\wh{\bU}+\ba\wh{\bS}^{(0,1)}=q\bU+\bS^{(1,0)}\ba-\bU\ba\wh{\bU},\\
&q\bU-\ba\bS^{(0,1)}=q\wh{\bU}-\wh{\bS}^{(1,0)}\ba+\wh{\bU}\ba\bU.
\end{align}
\end{subequations}
From these equations, we eliminate $\bS^{(0,1)}$ and $\bS^{(1,0)}$ to derive a closed‑form equation that involves only the matrix function $\bU$:
\begin{align}
\label{bU}
(p-q)(\wh{\wt{\bU}}-\bU)+(\wt{\bU}-\wh{\bU})((p+q)\bI+\ba\bU-\ba\wh{\wt{\bU}})=\bm 0.
\end{align}
Setting $(i,j)=(-1,0)$ in \eqref{Sij1-t1} and \eqref{Sij1-t2} and replacing $\bS^{(-1,1)}$ with equation \eqref{VW}, we have
\begin{subequations}
\begin{align}
\label{pV1}
&p(\wt{\bV}-\bV)=(\ba\wt{\bU}-\bU\ba)\wt{\bV},\\
\label{pV2}
&p(\wt{\bV}-\bV)=(\ba\bU-\wt{\bU}\ba)\bV.
\end{align}
\end{subequations}
Right‑multiplying equation \eqref{pV1} by $\wt{\bW}$ and equation \eqref{pV2} by $\bW$, and noting that $\bV\bW=\bI$, we get
\begin{subequations}
\begin{align}
p\bV\wt{\bW}=p\bI+\bU\ba-\ba\wt{\bU},\\
p\wt{\bV}\bW=p\bI+\ba\bU-\wt{\bU}\ba,
\end{align}
\end{subequations}
which imply
\begin{align}
(p\bI+\bU\ba-\ba\wt{\bU})(p\bI+\ba\bU-\wt{\bU}\ba)=p^2\bI.
\end{align}
Similarly, the evolution relation of $\bU$ in the $m$-direction is achieved by replacing $p$ with $q$ and the 
$\wt{\phantom{a}}$-shift with the $\wh{\phantom{a}}$-shift, leading to
\begin{subequations}
\label{th-UVW}
\begin{align}
&q\bV\wh{\bW}=q\bI+\bU\ba-\ba\wh{\bU},\\
&q\wh{\bV}\bW=q\bI+\ba\bU-\wh{\bU}\ba,\\
&(q\bI+\bU\ba-\ba\wh{\bU})(q\bI+\ba\bU-\wh{\bU}\ba)=q^2\bI.
\end{align}
\end{subequations}

To catch the dhGI model, the relations obtained above are far from sufficient. We must determine in detail the relationships between $U_{i}$ and $V_{i}$, 
which necessitates expanding some of the matrix relations mentioned above. Since $\bV\bW=\bI$ and $\det(\bV)=1$, 
each entry of the matrix $\bW$ defined in \eqref{UVW} can be explicitly expressed as
\begin{align}
\label{W}
\bW=\bI+\bS^{(0,-1)}=\begin{pmatrix}
V_{4}&-V_{2}\\
-V_{3}&V_{1}
\end{pmatrix}.
\end{align}
Under the specific definitions of $\bU,\bV$ in \eqref{UVW} and $\bW$ in \eqref{W}, we expand the 
matrix equation \eqref{th-UVW} to the scalar level and simplify it to obtain
\begin{subequations}
\label{h-UiVi}
\begin{align}
\label{h-UiVi1}
&V_{1}\wh{V}_{4}-V_{2}\wh{V}_{3}=\wh{V}_{1}V_{4}-\wh{V}_{2}V_{3},\\
\label{h-UiVi2}
&q+A_{1}=q(\wh{V}_{1}V_{4}-\wh{V}_{2}V_{3}),\\
\label{h-UiVi3}
&A_{2}=q(V_{1}\wh{V}_{2}-V_{2}\wh{V}_{1}),\\
\label{h-UiVi4}
&A_{3}=q(V_{3}\wh{V}_{4}-V_{4}\wh{V}_{3}),
\end{align}
\end{subequations}
and
\begin{align}
\label{Ai}
(q+A_{1})^2+A_{2}A_{3}=q^2,
\end{align}
where
\begin{align}
A_{1}:=U_{1}-\wh{U}_{1}=\wh{U}_{4}-U_{4},\quad A_{2}:=U_{2}+\wh{U}_{2},\quad A_{3}:=U_{3}+\wh{U}_{3}.
\end{align}
Similarly, we introduce a new variable
\begin{align}
\mV:=\mathrm{i}V_{2}/V_{4},
\end{align}
and revisit relation \eqref{h-UiVi}, which follows from substituting $-\mathrm{i}\mV V_{4}$ for $V_{2}$. Starting with \eqref{h-UiVi3}, we have
\begin{align}
\label{A2-1}
A_{2}=-\mathrm{i}q(V_{1}\wh{V}_{4}\wh{\mV}-\wh{V}_{1}V_{4}\mV).
\end{align}
Replacing $V_{1}\wh{V}_{4}$ in \eqref{A2-1} by \eqref{h-UiVi1} gives
\begin{align}
\label{A2-2}
A_{2}=q\big[\mathrm{i}\wh{V}_{1}V_{4}(\mV-\wh{\mV})-\mV\wh{\mV}\wh{V}_{3}V_{4}+\wh{\mV}^2V_{3}\wh{V}_{4}\big] .
\end{align}
Substituting $qV_{3}\wh{V}_{4}$ in \eqref{A2-2} with \eqref{h-UiVi4} then leads to
\begin{align}
\label{A2-3}
A_{2}=\mathrm{i}(\mV-\wh{\mV})(q\wh{V}_{1}V_{4}+\mathrm{i}q\wh{\mV}\wh{V}_{3}V_{4})+\wh{\mV}^2A_{3}.
\end{align}
Finally, replacing $q\wh{V}_{1}V_{4}$ in \eqref{A2-3} by \eqref{h-UiVi2} yields
\begin{align}
\label{Ai-V}
A_{2}=\mathrm{i}(\mV-\wh{\mV})(q+A_{1})+\mV\wh{\mV}A_{3},
\end{align}
which establishes the relation between $A_{i}~(i=1,2,3)$ and $\mV$. We take \eqref{Ai-V} into \eqref{Ai} to derive an equation for $(q+A_{1})$, and solving it indicates
\begin{align}
\label{A1}
2(q+A_{1})=\mathrm{i}A_{3}(\wh{\mV}-\mV)+\big[ 4q^2-A_{3}^2(\mV+\wh{\mV})^2\big]^{\frac{1}{2}}.
\end{align}
The symmetry between $(p,\wt{\phantom{a}})$ and $(q,\wh{\phantom{a}})$ allows us to obtain the relations for the $\wt{\phantom{a}}$-shift 
directly, as presented below
\begin{subequations}
\begin{align}
\label{Bi}
&(p+B_{1})^2+B_{2}B_{3}=p^2,\\
\label{Bi-V}
&B_{2}=\mathrm{i}(\mV-\wt{\mV})(p+B_{1})+\mV\wt{\mV}B_{3},\\
\label{B1}
&2(p+B_{1})=\mathrm{i}B_{3}(\wt{\mV}-\mV)+\big[4p^2-B_{3}^2(\mV+\wt{\mV})^2\big] ^{\frac{1}{2}},
\end{align}
\end{subequations}
where
\begin{align}
B_{1}:=U_{1}-\wt{U}_{1}=\wt{U}_{4}-U_{4},\quad B_{2}:=U_{2}+\wt{U}_{2},\quad B_{3}:=U_{3}+\wt{U}_{3}.
\end{align}

Next, utilizing the earlier relation connecting the $\wh{\phantom{a}}$-shift and the $\wt{\phantom{a}}$-shift derived previously, we construct  
closed dhGI models involving only the variables $U_{3}$ and $\mV$. To begin, we expand \eqref{bU} and consider its (2,1)-component
\begin{align}
\label{U3U1}
(p-q+\wh{U}_{1}-\wt{U}_{1})(\wh{\wt{U}}_{3}-U_{3})+(p+q+U_{1}-\wh{\wt{U}}_{1})(\wt{U}_{3}-\wh{U}_{3})=0.
\end{align}
Meanwhile, we replace equations \eqref{Ai-V} and \eqref{Bi-V} with the following identity
\begin{align}
A_{2}+\wt{A}_{2}=B_{2}+\wh{B}_{2},
\end{align}
and thereby obtain
\begin{align}
\label{VU1}
(p-q+\wh{U}_{1}-\wt{U}_{1}-\mathrm{i}\wt{\mV}\wt{U}_{3}+\mathrm{i}\wh{\mV}\wh{U}_{3})(\wh{\wt{\mV}}-\mV)+(p+q+U_{1}-\wh{\wt{U}}_{1}-\mathrm{i}\wh{\wt{\mV}}\wh{\wt{U}}_{3}+\mathrm{i}\mV U_{3})(\wt{\mV}-\wh{\mV})=0.
\end{align}
Note that equations \eqref{U3U1} and \eqref{VU1} contain several terms involving $U_{1}$. 
Therefore, our goal is to replace these terms with a relation that involves only $U_{3}$ and $\mV$. We find that the following identities hold 
\begin{subequations}
\label{U1-rela}
\begin{align}
p-q+\wh{U}_{1}-\wt{U}_{1}=
\label{U1-rela-a}
(p+B_{1})-(q+A_{1}),\\
\label{U1-rela-b}
=(p+\wh{B}_{1})-(q+\wt{A}_{1}), \\
p+q+U_{1}-\wh{\wt{U}}_{1}=
\label{U1-rela-c}
(p+\wh{B}_{1})+(q+A_{1}),\\
\label{U1-rela-d}
=(p+B_{1})+(q+\wt{A}_{1}).
\end{align}
\end{subequations}
With the aid of \eqref{A1} and \eqref{B1}, we construct four families of dhGI models, which are summarized below.

\noindent$\bullet$ If we set (\eqref{U1-rela-a}, \eqref{U1-rela-c}), then we have
\begin{subequations}
\label{dhGI1}
\begin{align}
\label{dhGI1-1}
& (U_{3}(\wt{\mV}-\wh{\mV})+ \mV(\wh{U}_{3}-\wt{U}_{3})-\wh{U}_{3}\wh{\mV}+\wt{U}_{3}\wt{\mV}-\mathrm{i}P^{'}+\mathrm{i}Q^{'}) (\wh{\wt{U}}_{3}-U_{3}) \nn \\
& \qquad +(\wh{U}_{3}(\wh{\wt{\mV}}-\mV)+\wh{\mV}(U_{3}-\wh{\wt{U}}_{3})+\wh{\wt{U}}_{3}\wh{\wt{\mV}}-U_{3}\mV-\mathrm{i}\wh{P}^{'}-\mathrm{i}Q^{'}) (\wt{U}_{3}-\wh{U}_{3})=0,\\
\label{dhGI1-2}
& (\mV(\wt{U}_{3}-\wh{U}_{3})+U_{3}(\wh{\mV}-\wt{\mV})-\wh{U}_{3}\wh{\mV}+\wt{U}_{3}\wt{\mV}+\mathrm{i}P^{'}-\mathrm{i}Q^{'}) (\wh{\wt{\mV}}-\mV) \nn \\
& \qquad +(\wh{\mV}(\wh{\wt{U}}_{3}-U_{3})+\wh{U}_{3}(\mV-\wh{\wt{\mV}})+\wh{\wt{U}}_{3}\wh{\wt{\mV}}-U_{3}\mV+\mathrm{i}\wh{P}^{'}+\mathrm{i}Q^{'}) (\wt{\mV}-\wh{\mV})=0,
\end{align}
\end{subequations}
where
\begin{align}
P^{'}:=[ 4p^2-(U_{3}+\wt{U}_{3})^2(\mV+\wt{\mV})^2]^{\frac{1}{2}},\quad Q^{'}:=[4q^2-(U_{3}+\wh{U}_{3})^2(\mV+\wh{\mV})^2]^{\frac{1}{2}}.
\end{align}
\noindent$\bullet$ If we set (\eqref{U1-rela-a}, \eqref{U1-rela-d}), then we have
\begin{subequations}
\label{dhGI2}
\begin{align}
\label{dhGI2-1}
&  (U_{3}(\wt{\mV}-\wh{\mV})+\mV(\wh{U}_{3}-\wt{U}_{3})-\wh{U}_{3}\wh{\mV}+\wt{U}_{3}\wt{\mV}-\mathrm{i}P^{'}+\mathrm{i}Q^{'}) (\wh{\wt{U}}_{3}-U_{3}) \nn \\
& \qquad +(\wt{U}_{3}(\wh{\wt{\mV}}-\mV)+\wt{\mV}(U_{3}-\wh{\wt{U}}_{3})+\wh{\wt{U}}_{3}\wh{\wt{\mV}}-U_{3}\mV-\mathrm{i}P^{'}-\mathrm{i}\wt{Q}^{'}) (\wt{U}_{3}-\wh{U}_{3})=0, \\
\label{dhGI2-2}
& (\mV(\wt{U}_{3}-\wh{U}_{3})+U_{3}(\wh{\mV}-\wt{\mV})-\wh{U}_{3}\wh{\mV}+\wt{U}_{3}\wt{\mV}+\mathrm{i}P^{'}-\mathrm{i}Q^{'}) (\wh{\wt{\mV}}-\mV) \nn \\
& \qquad +(\wt{\mV}(\wh{\wt{U}}_{3}-U_{3})+\wt{U}_{3}(\mV-\wh{\wt{\mV}})+\wh{\wt{U}}_{3}\wh{\wt{\mV}}-U_{3}\mV+\mathrm{i}P^{'}+\mathrm{i}\wt{Q}^{'}) (\wt{\mV}-\wh{\mV})=0.
\end{align}
\end{subequations}
\noindent$\bullet$~~If we set (\eqref{U1-rela-b}, \eqref{U1-rela-c}),  then we have
\begin{subequations}
\label{dhGI3}
\begin{align}
\label{dhGI3-1}
& (\wh{\wt{U}}_{3}(\wt{\mV}-\wh{\mV})+\wh{\wt{\mV}}(\wh{U}_{3}-\wt{U}_{3})-\wh{U}_{3}\wh{\mV}+\wt{U}_{3}\wt{\mV}-\mathrm{i}\wh{P}^{'}
+\mathrm{i}\wt{Q}^{'}) (\wh{\wt{U}}_{3}-U_{3}) \nn \\
&\qquad +(\wh{U}_{3}(\wh{\wt{\mV}}-\mV)+\wh{\mV}(U_{3}-\wh{\wt{U}}_{3})+\wh{\wt{U}}_{3}\wh{\wt{\mV}}-U_{3}\mV-\mathrm{i}\wh{P}^{'}-\mathrm{i}Q^{'}) (\wt{U}_{3}-\wh{U}_{3})=0, \\
\label{dhGI3-2}
& (\wh{\wt{\mV}}(\wt{U}_{3}-\wh{U}_{3})+\wh{\wt{U}}_{3}(\wh{\mV}-\wt{\mV})-\wh{U}_{3}\wh{\mV}+\wt{U}_{3}\wt{\mV}+\mathrm{i}\wh{P}^{'}-\mathrm{i}\wt{Q}^{'}) (\wh{\wt{\mV}}-\mV) \nn \\
&\qquad +(\wh{\mV}(\wh{\wt{U}}_{3}-U_{3})+\wh{U}_{3}(\mV-\wh{\wt{\mV}})+\wh{\wt{U}}_{3}\wh{\wt{\mV}}-U_{3}\mV+\mathrm{i}\wh{P}^{'}+\mathrm{i}Q^{'}) (\wt{\mV}-\wh{\mV})=0.
\end{align}
\end{subequations}
\noindent$\bullet$~~If we set (\eqref{U1-rela-b}, \eqref{U1-rela-d}), then we have
\begin{subequations}
\label{dhGI4}
\begin{align}
\label{dhGI4-1}
& (\wh{\wt{U}}_{3}(\wt{\mV}-\wh{\mV})+\wh{\wt{\mV}}(\wh{U}_{3}-\wt{U}_{3})-\wh{U}_{3}\wh{\mV}+\wt{U}_{3}\wt{\mV}-\mathrm{i}\wh{P}^{'}+\mathrm{i}\wt{Q}^{'})(\wh{\wt{U}}_{3}-U_{3}) \nn\\
& \qquad +(\wt{U}_{3}(\wh{\wt{\mV}}-\mV)+\wt{\mV}(U_{3}-\wh{\wt{U}}_{3})+\wh{\wt{U}}_{3}\wh{\wt{\mV}}-U_{3}\mV-\mathrm{i}P^{'}-\mathrm{i}\wt{Q}^{'}) (\wt{U}_{3}-\wh{U}_{3})=0, \\
\label{dhGI4-2}
& (\wh{\wt{\mV}}(\wt{U}_{3}-\wh{U}_{3})+\wh{\wt{U}}_{3}(\wh{\mV}-\wt{\mV})-\wh{U}_{3}\wh{\mV}+\wt{U}_{3}\wt{\mV}+\mathrm{i}\wh{P}^{'}-\mathrm{i}\wt{Q}^{'}) (\wh{\wt{\mV}}-\mV) \nn \\
& \qquad +(\wt{\mV}(\wh{\wt{U}}_{3}-U_{3})+\wt{U}_{3}(\mV-\wh{\wt{\mV}})+\wh{\wt{U}}_{3}\wh{\wt{\mV}}-U_{3}\mV+\mathrm{i}P^{'}+\mathrm{i}\wt{Q}^{'}) (\wt{\mV}-\wh{\mV})=0.
\end{align}
\end{subequations}

The models \eqref{dhGI1}, \eqref{dhGI2}, \eqref{dhGI3} and \eqref{dhGI4} presented above are the dhGI systems. 
Their continuum limits yield the continuous hGI system (see Subsection \ref{4-4}). Their soliton and multipole solutions will be provided in Subsection \ref{4-3}. 
Moreover, since these models are conjugate symmetric, their local reductions will be discussed. By recombining the lattice equations in these models, 
one finds that the combined models admit nonlocal reductions, which constitutes the main focus of Subsection \ref{4-5}.

\begin{remark}
Alternatively, we may use a pair of new dependent variables $(\mathrm{i}V_{3}/V_{1},U_{2})$ to construct a dhGI model identical to \eqref{dhGI1} (or \eqref{dhGI2}, \eqref{dhGI3}, \eqref{dhGI4}). This indicates that $(\mathrm{i}V_{3}/V_{1},U_{2})$ are also solutions of the discrete GI models \eqref{dhGI1}, \eqref{dhGI2}, \eqref{dhGI3}, \eqref{dhGI4}, where $\mathrm{i}V_{3}/V_{1}$ corresponds to $U_{3}$ and $U_{2}$ to $\mV$ in the original models. In the following study, we present only the solution expressed in terms of the dependent variables $(U_{3},\mV)$ without loss of generality.
\end{remark}

\subsection{Solutions to the dhGI models}\label{4-3}

We next solve SDE-II and thus construct explicit solutions of the dhGI models. As in the dGI case (Subsection \ref{3-3}), 
the Cauchy matrix framework converts the solution problem into the choice of spectral matrices. Diagonal spectral matrices yield soliton solutions, 
while Jordan-block spectral matrices generate multiple-pole solutions. The solution formulae share the same structural form \eqref{sym-uv-def} 
as in the dGI case, where essential difference lies in the plane wave factors, which are now determined by SDE-II rather than SDE-I. Under the setting 
of \eqref{KMrc}, we first expand SDE-II into
\begin{subequations}
\label{SDE-II-1}
\begin{align}
&\bK_{1}\bM_{1}-\bM_{1}\bK_{2}=\br_{1}\tc_{2},\quad \bK_{2}\bM_{2}-\bM_{2}\bK_{1}=\br_{2}\tc_{1}, \\
&(p\bI-\bK_{1})\wt{\br}_{1}=(p\bI+\bK_{1})\br_{1},\quad (p\bI+\bK_{2})\wt{\br}_{2}=(p\bI-\bK_{2})\br_{2},\\
&(q\bI-\bK_{1})\wh{\br}_{1}=(q\bI+\bK_{1})\br_{1},\quad (q\bI+\bK_{2})\wh{\br}_{2}=(q\bI-\bK_{2})\br_{2}.
\end{align}
\end{subequations}
The matrix functions $\bU$ and $\bV$, as defined in \eqref{UVW}, can be expressed as
\begin{align}
\bU=\bS^{(0,0)}=\tc(\bI+\bM)^{-1}\br,\quad\bV=\bI-\bS^{(-1,0)}
=\bI-\tc(\bI+\bM)^{-1}\bK^{-1}\br.
\end{align}
Since the solution of the dhGI model consists of the functions $(U_{3},\mV)$, we need to obtain their explicit expressions. 
With the aid of \eqref{si}, these are given by
\begin{align}
\label{U3V-def}
(U_{3},\mV)=(U_{3},\mathrm{i}V_{2}/V_{4}) =\left( \tc_{1}(\bI-\bM_{1}\bM_{2})^{-1}\br_{1},
\frac{-\mathrm{i}\tc_{2}(\bI-\bM_{2}\bM_{1})^{-1}\bK_{2}^{-1}\br_{2}}{1+\tc_{1}
\bM_{1}(\bI-\bM_{2}\bM_{1})^{-1}\bK_{2}^{-1}\br_{2}}\right),
\end{align}
where the relevant elements are defined in \eqref{SDE-II-1}.

Next, we present the explicit formulae for soliton and multiple-pole solutions by solving \eqref{SDE-II-1}.
Due to the similarity invariance of $\bS^{(i,j)}$ (see Proposition \ref{prop-2-3}), we consider only the cases where 
$\bK_{1}$ and $\bK_{2}$ are diagonal matrices or Jordan matrices, corresponding to soliton solutions and multiple-pole solutions, respectively.

\begin{theorem}[\textbf{Soliton solutions}]
Functions $(U_{3},\mV)$ defined in \eqref{U3V-def} serve as $N$-soliton solutions for the dhGI model \eqref{dhGI1}(or \eqref{dhGI2}-\eqref{dhGI4}) with
\begin{subequations}
\begin{align}
&\bK_{1}=\mathrm{diag}(k_{1},k_{2},\cdots,k_{N}), \quad \bK_{2}=\mathrm{diag}(l_{1},l_{2},\cdots,l_{N}),\\
&\br_{1}=(\phi(k_{1}),\phi(k_{2}),\cdots,\phi(k_{N}))^{\st}, \quad \br_{2}=(\psi(l_{1}),\psi(l_{2}),\cdots,\psi(l_{N}))^{\st},\\
&\tc_{1}=(c_{1,1},c_{1,2},\cdots,c_{1,N}), \quad \tc_{2}=(c_{2,1},c_{2,2},\cdots,c_{2,N}),\\
&\bM_{1}=(M_{1,ij})_{N\times N},\quad M_{1,ij}=\frac{\phi(k_{i})c_{2,j}}{k_{i}-l_{j}},\\
&\bM_{2}=(M_{2,ij})_{N\times N},\quad M_{2,ij}=\frac{\psi(l_{i})c_{1,j}}{l_{i}-k_{j}},
\end{align}
\end{subequations}
where
\begin{subequations}
\label{phi-psi}
\begin{align}
&\phi(z)=\left(\frac{p+z}{p-z}\right)^{n}
\left(\frac{q+z}{q-z}\right) ^{m}\phi^{0}(z),\\
&\psi(z)=\left(\frac{p-z}{p+z}\right)^{n}
\left(\frac{q-z}{q+z}\right)^{m}\psi^{0}(z).
\end{align}
\end{subequations}
\end{theorem}
\begin{theorem}[\textbf{Multiple-pole solutions}]
Let $(U_{3},\mV)$ be defined by \eqref{U3V-def}. Then they constitute $N$-th order multiple-pole solutions of the dhGI model \eqref{dhGI1} if the parameters are chosen according to either of the following two cases.

\noindent$\bullet$~~Case-I.
\begin{subequations}
\begin{align}
&\bK_{1}=\bJ[k,1],\quad\bK_{2}=\bJ[l,1],\quad \mathrm{with}\quad k\neq l,\\
&\br_{1}=\bF_{k}[\phi(k)]\be_{N},\quad\br_{2}=\bF_{l}[\psi(l)]\be_{N},\\
&\tc_{1}=\be^{\st}_{N}\bH[c_{1,j}],\quad\tc_{2}=\be^{\st}_{N}\bH[c_{2,j}],\\
&\bM_{1}=\bF_{k}[\phi(k)]\bG_{1}\bH[c_{2,j}],\quad\bM_{2}=\bF_{l}[\psi(l)]\bG_{2}\bH[c_{1,j}],
\end{align}
\end{subequations}
where $\bG_{1},\bG_{2}$ are given in \eqref{G1G2} and $\bJ,\bF,\bH$ are defined in \eqref{J-def}, \eqref{F-def} and \eqref{H-def}.

$\bullet$~~Case-II.
\begin{subequations}
\begin{align}
&\bK_{1}=\bJ[k,1],\quad\bK_{2}=\bJ[l,-1],\quad \mathrm{with}\quad k\neq l,\\
&\br_{1}=\bF_{k}[\phi(k)]\be_{N},\quad\br_{2}=\bF_{-l}[\psi(l)]\be_{N},\\
&\tc_{1}=\be^{\st}_{N}\bH[c_{1,j}],\quad\tc_{2}=\be^{\st}_{N}\bH[c_{2,j}],\\
&\bM_{1}=\bF_{k}[\phi(k)]\bG^{'}\bH[c_{2,j}],\quad\bM_{2}=-\bF_{-l}[\psi(l)]\bG^{'}\bH[c_{1,j}],\\
&\bG^{'}=(G^{'}_{i,j})_{N\times N},\quad G^{'}_{i,j}=\frac{(-1)^{i+j}C^{i-1}_{i+j-2}}{(k-l)^{i+j-1}}.
\end{align}
\end{subequations}
\end{theorem}

\subsection{Continuum limits: from dhGI to the hGI system}\label{4-4}

We now study the continuum limits of the dhGI models. The two-step limiting procedure employed here mirrors that of 
Subsection 3.4: first one lattice direction is contracted to obtain a semi-discrete model, then the second direction 
is contracted to reach the continuous system. However, a crucial difference arises: in the dGI case the coordinate 
transformation involves $\partial_\tau = \frac{1}{2a}\partial_t + \partial_x$, which produces the second-order GI system, 
whereas in the present case the transformation $\partial_\tau = \frac{1}{12p^2}\partial_t + \partial_x$ naturally gives rise 
to the third-order hGI system. This difference in the scaling reflects the distinct dispersive orders encoded in SDE-I and SDE-II.

\subsubsection{The sdhGI model}

In the first step we let both $q$ and $m$ tend to infinity but keep $m/q$ finite. Introduce
\begin{equation}
\xi=m/q,
\end{equation}
where $1/q$ serves as the grid parameter in the discretization of $\xi$. With this setting
we reinterpret the dependent variables $U_{3}$ and $\mV$ as
\begin{align}
U_{3}(n,m)=U_{3}(n,\xi),\quad\mV(n,m)=\mV(n,\xi),
\end{align}
without making confusions,
while the shifted variables in $m$-direction give rise to
\begin{subequations}
\begin{align}
\wh{U}_{3}=&U_{3}(n,\xi+1/q)=U_{3}+U_{3,\xi}/q+U_{3,\xi\xi}/(2q^{2})+\cdots,\\
\wh{\wt{U}}_{3}=&U_{3}(n+1,\xi+1/q)=\wt{U}_{3}+\wt{U}_{3,\xi}/q+\wt{U}_{3,\xi\xi}/(2q^{2})+\cdots,
\end{align}
\end{subequations}
and a similar formula  for $\wh{\mV}$ and $\wh{\wt{\mV}}$.
Inserting them into the system \eqref{U3U1} and \eqref{VU1}, then
from the leading term (in terms of $\mathcal{O}(q^0)$) we obtain a semi-discrete system
\begin{subequations}
\label{sd-U3VU1}
\begin{align}
&U_{3,\xi}+\wt{U}_{3,\xi}+2(p+U_{1}-\wt{U}_{1})(U_{3}-\wt{U}_{3})=0,\\
&\mV_{\xi}+\wt{\mV}_{\xi}+2(p+U_{1}-\wt{U}_{1}+\mathrm{i}U_{3}\mV-\mathrm{i}\wt{U}_{3}\wt{\mV})(\mV-\wt{\mV})=0.
\end{align}
\end{subequations}
It is observed that this semi‑discrete model also contains a term involving $U_{1}$, which we need to eliminate. With the aid of equations \eqref{B1}, 
the semi‑discrete model \eqref{sd-U3VU1} simplifies to a form that involves only $U_{3}$ and $\mV$:
\begin{subequations}
\label{sd-U3V}
\begin{align}
&U_{3,\xi}+\wt{U}_{3,\xi}+(P^{'}+\mathrm{i}(U_{3}+\wt{U}_{3})(\wt{\mV}-\mV))(U_{3}-\wt{U}_{3})=0,\\
&\mV_{\xi}+\wt{\mV}_{\xi}+(P^{'}-\mathrm{i}(\mV+\wt{\mV})(\wt{U}_{3}-U_{3}))(\mV-\wt{\mV})=0,
\end{align}
where
\begin{align}
P^{'}:=P^{'}(n,\xi)=\big[ 4p^2-(U_{3}+\wt{U}_{3})^2(\mV+\wt{\mV})^2\big]^{\frac{1}{2}}.
\end{align}
\end{subequations}
We refer to the system \eqref{sd-U3V} as the sdhGI model. To obtain its solutions, we examine how the above limit affects $\phi(z)$ and $\psi(z)$ 
as defined in \eqref{phi-psi}. Under this limit, $\phi(z)$ and $\psi(z)$ reduce to
\begin{align}
\label{sd-phi}
\phi(z)=\left(\frac{p+z}{p-z}\right)^{n}
e^{2z\xi}\phi^{0}(z),\quad
\psi(z)=\left(\frac{p-z}{p+z}\right)^{n}
e^{-2z\xi}\psi^{0}(z).
\end{align}
In this way, the solution of the sdhGI model \eqref{sd-U3V} can also be obtained via the Cauchy matrix scheme, leading to the following proposition.
\begin{prop}
Formulae \eqref{U3V-def} provide solutions to the sdhGI
model \eqref{sd-U3V}, through the set of determining equations
\begin{subequations}
\begin{align}
&\bK_{1}\bM_{1}-\bM_{1}\bK_{2}=\br_{1}\tc_{2},\quad \bK_{2}\bM_{2}-\bM_{2}\bK_{1}=\br_{2}\tc_{1}, \\
&(p\bI-\bK_{1})\wt{\br}_{1}=(p\bI+\bK_{1})\br_{1},\quad \br_{1,\xi}=2\bK_{1}\br_{1},\\
&(p\bI+\bK_{2})\wt{\br}_{2}=(p\bI-\bK_{2})\br_{2},\quad \br_{2,\xi}=-2\bK_{2}\br_{2}.
\end{align}
\end{subequations}
\end{prop}

\subsubsection{Continuous hGI system}
Now, we apply the second continuum limit to the sdhGI model \eqref{sd-U3V} to obtain the continuous hGI system. To this end, we set
\begin{align}
n\rightarrow\infty,\quad p\rightarrow\infty,\quad\text{while}\quad \tau:=n/p  \quad\text{finite},
\end{align}
along with
\begin{align}
U_{1}(n,\xi)=U_{1}(\tau,\xi),\quad U_{3}(n,\xi)=U_{3}(\tau,\xi),\quad\mV(n,
\xi)=\mV(\tau,\xi),
\end{align}
where $\tau$ is another continuous coordinate. To obtain the nontrivial system, 
we introduce new coordinates defined in terms of their derivatives
\begin{align}
\partial_{\tau}=\frac{1}{12p^2}\partial_{t}+\partial_{x},\quad\partial_{\xi}=\partial_{x}.
\end{align}
We may expand
\begin{align}
\begin{aligned}
\wt{U}_{3}=&U_{3}+U_{3,\tau}/p+U_{3,\tau\tau}/(2p^2)+U_{3,\tau\tau\tau}/(3p^3)+\cdots\\
=&U_{3}+U_{3,x}/p+U_{3,xx}/(2p^2)+(U_{3,t}+2U_{3,xxx})/(12p^3)+\cdots,
\end{aligned}
\end{align} 
and a similar formula  for $\wt{U}_{1}$ and $\wt{\mV}$. Inserting them into \eqref{B1} and \eqref{sd-U3VU1}, 
the leading order (in terms of $\mathcal{O}(p^{-1})$ and $\mathcal{O}(p^{-2})$) yield a coupled system
\begin{subequations}
\begin{align}
&U_{1,x}+\mathrm{i}U_{3}\mV_{x}-2U_{3}^2\mV^2=0,\\
&U_{3,t}-U_{3,xxx}-12U_{3,x}U_{1,x}=0,\\
&\mV_{t}-\mV_{xxx}-12\mV_{x}(U_{1,x}+\mathrm{i}U_{3}\mV_{x}+\mathrm{i}U_{3,x}\mV)=0,
\end{align}
\end{subequations}
which imply
\begin{subequations}
\label{ctn-U3V}
\begin{align}
&U_{3,t}-U_{3,xxx}+12\mathrm{i}U_{3}U_{3,x}\mV_{x}-24U_{3}^2\mV^2U_{3,x}=0,\\
&\mV_{t}-\mV_{xxx}-12\mathrm{i}\mV\mV_{x}U_{3,x}-24U_{3}^2\mV^2\mV_{x}=0.
\end{align}
\end{subequations}
The system \eqref{ctn-U3V} is the well‑known hGI system.

Solving the hGI system \eqref{ctn-U3V} requires studying the variation of the plane wave factor \eqref{sd-phi}. We expand \eqref{sd-phi} as
\begin{align}
\begin{aligned}
\phi(z)=&\left(\frac{p+z}{p-z}\right)^{n}e^{2z\xi}\phi^{0}(z)\\
=&\exp\left\lbrace 2z\xi+p\tau\left[ \ln\left(1+ \frac{z}{p}\right)-\ln\left(1- \frac{z}{p}\right)\right]  \right\rbrace\phi^{0}(z) \\
=&\exp\left\lbrace 2z(\xi+\tau)+\frac{2z^3\tau}{3p^2}+\tau\mathcal{O}(p^{-4})\right\rbrace\phi^{0}(z) \\
=&\exp\left\lbrace 2zx+8z^3t+\mathcal{O}(p^{-2})\right\rbrace\phi^{0}(z),
\end{aligned}
\end{align}
an analogous expansion form for $\psi(z)$ is given by
\begin{align}
\psi(z)=\exp\left\lbrace -2zx-8z^3t+\mathcal{O}(p^{-2})\right\rbrace\psi^{0}(z).
\end{align}
Consequently, solutions for the hGI system \eqref{ctn-U3V} can also be obtained from
the Cauchy matrix scheme.
\begin{prop}
Formulae \eqref{U3V-def} provide solutions to the hGI
system \eqref{ctn-U3V}, through the set of determining equations
\begin{subequations}
\begin{align}
&\bK_{1}\bM_{1}-\bM_{1}\bK_{2}=\br_{1}\tc_{2},\quad \bK_{2}\bM_{2}-\bM_{2}\bK_{1}=\br_{2}\tc_{1}, \\
&\br_{1,x}=2\bK_{1}\br_{1},\quad\br_{1,t}=8\bK_{1}^3\br_{1},\\
&\br_{2,x}=2\bK_{2}\br_{2},\quad\br_{2,t}=8\bK_{2}^3\br_{2}.
\end{align}
\end{subequations}
\end{prop}

\subsection{Local and nonlocal reductions}\label{4-5}

Finally, we investigate reductions of the dhGI models. In Subsection 3.5, the dGI models only admitted the local conjugate reduction 
$\nu = u_3^*$ (with nonlocal reductions being absent), which reduced each system to a scalar dGI equation. The present higher-order case 
is considerably richer. The local reduction $\mV=U_{3}^*$ still applies and produces scalar dhGI equations. More importantly, unlike the 
dGI models, suitable pairwise recombinations of the dhGI equations allow nonlocal reductions of the form $\mV = \mathrm{i}U_{3,-1}^*$\footnote{For a given discrete
function $f:=f_{n,m}$, we define $f_{-1}:= f_{-n,-m}, \wt{f}_{-1}:= f_{-n-1,-m}, \wh{f}_{-1}:= f_{-n,-m-1}$, and $\wh{\wt{f}}_{-1}:=f_{-n-1,-m-1}$.},
involving evaluations at reflected lattice points. This additional reduction mechanism is a structural feature specific to the 
higher-order hierarchy: it originates from the real-valued lattice parameters and the symmetric role of the sign matrix $\ba$ in SDE-II, 
which together create the algebraic conditions necessary for nonlocal symmetries to emerge.

\subsubsection{Local reduction}

For the dhGI model \eqref{dhGI1}, we observe that setting $\mV=U_{3}^*$ reduces the two equations in the model to a single equation, 
which we refer to as the dhGI equation. This reduction can also be realized at the solution level by imposing appropriate constraints. 
Models \eqref{dhGI2}, \eqref{dhGI3}, and \eqref{dhGI4} remain valid under this reduction. As a representative example, the reduction of model \eqref{dhGI1} gives
\begin{align}
\label{eq-hGI-1}
& (U_{3}(\wt{U}_{3}^*-\wh{U}_{3}^*)+U_{3}^*(\wh{U}_{3}-\wt{U}_{3})-|\wh{U}_{3}|^2+|\wt{U}_{3}|^2-\mathrm{i}P_{1}^{'}+\mathrm{i}Q_{1}^{'}) (\wh{\wt{U}}_{3}-U_{3}) \nn \\
& \qquad +(\wh{U}_{3}(\wh{\wt{U}}_{3}^*-U_{3}^*)+\wh{U}_{3}^*(U_{3}-\wh{\wt{U}}_{3})+|\wh{\wt{U}}_{3}|^2-|U_{3}|^2
-\mathrm{i}\wh{P}^{'}_{1}-\mathrm{i}Q^{'}_{1}) (\wt{U}_{3}-\wh{U}_{3})=0,
\end{align}
where
\begin{align}
\label{P'Q'-1}
P^{'}_{1}=[4p^2-|U_{3}+\wt{U}_{3}|^4]^{\frac{1}{2}},\quad
Q^{'}_{1}=[4q^2-|U_{3}+\wh{U}_{3}|^4]^{\frac{1}{2}}.
\end{align}
The remaining three dhGI equations, obtained from models \eqref{dhGI2}--\eqref{dhGI4} under the same reduction, are listed in Appendix \ref{App-C}.
We demonstrate the validity of this reduction at the solution level in the following theorem.
\begin{theorem}
The solution \eqref{U3V-def} admits the conjugate reduction $\mV=U_{3}^{*}$ under constraints
\begin{align}
\bK_{2}=-\bK_{1}^{*},\quad \br_{2}=\mathrm{i}\bK_{2}\br_{1}^{*},\quad \tc_{2}=\tc_{1}^{*}.
\end{align}
\end{theorem}
See the proof of Theorem \ref{T-3-3} for a similar argument.
The following theorem presents soliton and multipole solutions to the dhGI equation \eqref{eq-hGI-1}.
\begin{theorem}
The function 
\begin{align}
U_{3}=\tc_{1}(\bI-\mathrm{i}\bM_{1}\bK^{*}_{1}\bM^{*}_{1})^{-1}\br_{1},
\end{align}
provides i)  $N$-soliton solutions for the dhGI equation \eqref{eq-hGI-1} with
\begin{subequations}
\begin{align}
&\bK_{1}=\mathrm{diag}(k_{1},k_{2},\cdots,k_{N}),\quad k_{i}\in\mathbb{C}, \quad i=1,2,\cdots,N,\\
&\br_{1}=(\phi(k_{1}),\phi(k_{2}),\cdots,\phi(k_{N}))^{\st},\quad\tc_{1}=(c_{1,1},c_{1,2},\cdots,c_{1,N}),\\
&\bM_{1}=(M_{1,ij})_{N\times N},\quad M_{1,ij}=\frac{\phi(k_{i})c_{1,j}^*}{k_{i}+k_{j}^{*}};
\end{align}
\end{subequations}
and ii) the multiple-pole solutions for the dhGI equation \eqref{eq-hGI-1} with
\begin{subequations}
\begin{align}
&\bK_{1}=\bJ[k,1],\quad\br_{1}=\bF_{k}[\phi(k)]\be_{N},\quad\tc_{1}=\be_{N}^{\st}\bH[c_{1,j}],\\
&\bM_{1}=\bF_{k}[\phi(k)]\bG\bH^{*}[c_{1,j}], \quad
\bG=(G_{i,j})_{N\times N},\quad G_{i,j}=\frac{(-1)^{i-1}C^{i-1}_{i+j-2}}{(k+k^{*})^{i+j-1}},
\end{align}
\end{subequations}
where the discrete plane wave factor $\phi(z)$ is defined in \eqref{phi-psi},
and the matrices $\bJ,\bF,\bH$ are defined in \eqref{J-def}, \eqref{F-def} and \eqref{H-def}.
\end{theorem}

\subsubsection{Nonlocal reduction}

We recombine the lattice equations in models \eqref{dhGI1}, \eqref{dhGI2}, \eqref{dhGI3} and \eqref{dhGI4} pairwise so 
that the resulting four new models admit nonlocal reductions. After careful consideration, we pair \eqref{dhGI1-1} 
with \eqref{dhGI4-2}, \eqref{dhGI2-1} with \eqref{dhGI3-2}, \eqref{dhGI3-1} with \eqref{dhGI2-2}, and \eqref{dhGI4-1} with \eqref{dhGI1-2}, 
thereby obtaining four new families of dhGI models, which we denote as $(N.1)$, $(N.2)$, $(N.3)$, and $(N.4)$, respectively. 
These models themselves admit nonlocal reduction 
\begin{align}
\mV=\mathrm{i}U_{3,-1}^*,
\end{align}
then from $(N.1)-(N.4)$ one gets the four distinct nonlocal dhGI equations. As a representative example, 
the reduction of model $(N.1)$ gives
\begin{align}
\label{non-eq-hGI-1}
\begin{aligned}
&(U_{3}(\wt{U}_{3,-1}^*-\wh{U}_{3,-1}^*)+U_{3,-1}^*(\wh{U}_{3}-\wt{U}_{3})-\wh{U}_{3}\wh{U}_{3,-1}^*+\wt{U}_{3}\wt{U}_{3,-1}^*-P_{2}^{'}+Q_{2}^{'}) (\wh{\wt{U}}_{3}-U_{3})\\&+
(\wh{U}_{3}(\wh{\wt{U}}_{3,-1}^*-U_{3,-1}^*)+\wh{U}_{3,-1}^*(U_{3}-\wh{\wt{U}}_{3})+\wh{\wt{U}}_{3}\wh{\wt{U}}_{3,-1}^*-U_{3}U_{3,-1}^*-\wh{P}^{'}_{2}-Q^{'}_{2}) (\wt{U}_{3}-\wh{U}_{3})=0,
\end{aligned}
\end{align}
where
\begin{align}
\label{P'Q'-2}
P^{'}_{2}=[4p^2+(U_{3}+\wt{U}_{3})(U_{3,-1}^*+\wt{U}_{3,-1}^*)]^{\frac{1}{2}},\quad
Q^{'}_{2}=[4q^2+(U_{3}+\wh{U}_{3})(U_{3,-1}^*+\wh{U}_{3,-1}^*)]^{\frac{1}{2}}.
\end{align}
The remaining three nonlocal dhGI equations are listed in Appendix \ref{App-D}.

\begin{theorem}
The solution \eqref{U3V-def} admits the conjugate nonlocal reduction $\mV=\mathrm{i}U_{3,-1}^*,$ under constraints
\begin{align}
\bK_{2}=\bK_{1}^{*},\quad \br_{2}=-\bK_{2}\br_{1,-1}^{*},\quad \tc_{2}=\tc_{1}^{*}.
\end{align}
\end{theorem}
See the proof of Theorem \ref{T-3-3} for a similar argument.
The following theorem presents soliton and multiple-pole solutions to the nonlocal dhGI equation \eqref{non-eq-hGI-1}.
\begin{theorem}
The function 
\begin{align}
U_{3}=\tc_{1}(\bI+\bM_{1}\bK^{*}_{1}\bM^{*}_{1,-1})^{-1}\br_{1},
\end{align}
provides i)  $N$-soliton solutions for the nonlocal dhGI equation \eqref{non-eq-hGI-1} with
\begin{subequations}
\begin{align}
&\bK_{1}=\mathrm{diag}(k_{1},k_{2},\cdots,k_{N}),\quad k_{i}\in\mathbb{C}, \quad i=1,2,\cdots,N,\\
&\br_{1}=(\phi(k_{1}),\phi(k_{2}),\cdots,\phi(k_{N}))^{\st},\quad\tc_{1}=(c_{1,1},c_{1,2},\cdots,c_{1,N}),\\
&\bM_{1}=(M_{1,ij})_{N\times N},\quad M_{1,ij}=\frac{\phi(k_{i})c_{1,j}^*}{k_{i}-k_{j}^{*}};
\end{align}
\end{subequations}
and ii) the multiple-pole solutions for the nonlocal dhGI equation \eqref{non-eq-hGI-1} with
\begin{subequations}
\begin{align}
&\bK_{1}=\bJ[k,1],\quad\br_{1}=\bF_{k}[\phi(k)]\be_{N},\quad\tc_{1}=\be_{N}^{\st}\bH[c_{1,j}],\\
&\bM_{1}=\bF_{k}[\phi(k)]\bG\bH^{*}[c_{1,j}], \quad
\bG=(G_{i,j})_{N\times N},\quad G_{i,j}=\frac{(-1)^{i-1}C^{i-1}_{i+j-2}}{(k-k^{*})^{i+j-1}},
\end{align}
\end{subequations}
where the discrete plane wave factor $\phi(z)$ is defined in \eqref{phi-psi},
and the matrices $\bJ,\bF,\bH$ are defined in \eqref{J-def}, \eqref{F-def} and \eqref{H-def}, respectively.
\end{theorem}

\section{Conclusions}\label{5}

In this paper, we have systematically constructed discrete integrable analogues of the GI system and the hGI system by means of the Cauchy matrix framework. 
The main results of this work are summarized as follows.

Starting from the Sylvester equation with block matrix structures, we introduced two distinct sets of dispersion relations, which define different 
discrete dispersion relations for the plane wave vectors. Within SDE-I, we derived the shift dynamics of 
the master functions and, by specializing to low-order cases and eliminating auxiliary variables, 
we obtained closed lattice systems involving only two scalar dependent variables $(\mu,\nu)=(u_3, v_2/v_4)$. 
A notable feature of this construction is the non-uniqueness of the elimination step: since the expressions for the shift differences 
$(u_1-\wh{\wt{u}}_1)$ and $(\wt{u}_1-\wh{u}_1)$ admit multiple valid algebraic representations, 
we derived four conjugate-symmetric families of dGI models. Similarly, within SDE-II, we constructed four families of 
dhGI models with dependent variables $(U_3, \mV)$. This multiplicity of discretizations for a single continuous equation 
is an intrinsic phenomenon associated with derivative-type integrable systems and reflects the richness of the underlying algebraic structure.

For all discrete models obtained, we provided explicit exact solutions via the Cauchy matrix method. When the spectral 
matrices $\bK_1$ and $\bK_2$ are chosen as diagonal matrices, the construction yields $N$-soliton solutions expressed 
through discrete plane wave factors. When Jordan-block spectral matrices are employed, $N$-th order multiple-pole solutions
are obtained, expressed in terms of lower triangular Toeplitz matrices and skew-triangular Hankel matrices. The similarity invariance of the master 
functions ensures that these two canonical choices exhaust all relevant cases.

We verified the consistency of all discrete models through a two-step continuum limit procedure. In the first step, one lattice direction is 
contracted, reducing each fully discrete model to a sdGI (or sdhGI) model. In the second step, the remaining lattice direction is contracted to yield 
the continuous GI system \eqref{ctn-u3v} or the continuous hGI system \eqref{ctn-U3V}. Crucially, all four dGI models share the same continuum limit, 
namely the GI equation, and all four dhGI models converge to the same hGI equation. This universality of the continuous limit, despite the different 
discrete forms, confirms that the four models in each family are genuinely distinct discretizations of the same continuous integrable equation.

We also investigated symmetry reductions at both the equation level and the solution level. The conjugate symmetry of all discrete models allows 
the imposition of the local reduction $\nu = u_3^*$ (for dGI) or $\mV = U_3^*$ (for dhGI), which reduces each two-component system to a scalar equation. 
At the solution level, this reduction is realized by explicit constraints on the spectral data, namely $\bK_2 = \bK_1^*$ with appropriate conditions on the 
vectors $\br_2$ and $\tc_2$. In the higher-order case, we further discovered that pairwise recombinations of the dhGI lattice equations, obtained by pairing 
one equation from one model with one from another, admit nonlocal reductions of the form $\mV=\mathrm{i}U_{3,-1}^*$. These nonlocal reductions produce nonlocal dhGI 
equations together with their explicit soliton and multiple-pole solutions. This nonlocal reduction mechanism is specific to the higher-order hierarchy and is
absent in the dGI case, highlighting a fundamental structural difference between the two levels of the GI hierarchy.

Several directions for future research emerge naturally from this work. First, it would be valuable to derive Lax pairs for the discrete models 
obtained here and to investigate their integrability properties, such as B\"{a}cklund transformations and conservation laws, directly at the lattice 
level. Second, the non-uniqueness phenomenon observed in the discretization process raises the question of whether 
there exists a classification principle, perhaps based on symmetry or variational structure, that distinguishes among the four discrete models in each family 
\cite{HJN2016}. Third, extending the present approach to other derivative-type integrable equations, such as the Kaup--Newell \cite{KN1978} 
and Chen--Lee--Liu \cite{CLL1979} equations, could reveal whether similar multiplicities arise in their discretizations. 
Finally, the nonlocal reductions discovered here motivate a deeper study of nonlocal integrable lattice equations and their physical 
applications \cite{AblowitzMusslimani2013,ZhangZhao2022nonlocal}.

\vskip 20pt
\section*{Acknowledgments}
This project is supported by National Natural Science Foundation of China (No. 12071432) 
and Natural Science Foundation of Zhejiang Province (No. LZ24A010007).

\vskip 20pt
\section*{Data Availability Statement}
Data sharing not applicable to this article as no datasets were generated or
analyzed during the current study.

\vskip 20pt
\section*{Conflict of interest}
There are no conflicts of interest to declare.

\begin{appendix}

\titleformat{\section}{\Large\bfseries }{Appendix \thesection }{0.5em}{ }{}

\section{KP-type Cauchy matrix scheme for the dhGI models}

More precisely, the dhGI models in the main text are obtained via the AKNS-type Cauchy matrix scheme. However, the same models can also be derived 
using the KP-type Cauchy matrix scheme. Below, we present this KP-type construction in the form of a proposition.
\begin{prop}
Assume that $\bK,\bL$ are  invertible constant matrices in $\mathbb{C}^{N\times N}$. Let
\begin{align}
\bM=\bM_{1}+\bM_{2},\quad \br=(\br_{1},\br_{2}),\quad \bs=(\bs_{1},\bs_{2}),
\end{align}
where $\bM_{\iota}\in \mathbb{C}^{N\times N}[n,m]$,
$\br_{\iota}, \bs_{\iota} \in \mathbb{C}^{N\times 1}[n,m]$ are matrix functions of $(n,m)\in \mathbb{Z}^2$. With $\ba=\mathrm{diag}(1,-1)$, suppose that
\begin{subequations}
\begin{align}
&\bK\bM-\bM\bL=\br\bs^{\st},\\
&p\wt{\br}=p\br+\bK\br\ba,\quad q\wh{\br}=q\br+\bK\br\ba,\\
&p\bs=p\wt{\bs}+\bL^{\st}\wt{\bs}\ba,\quad q\bs=q\wh{\bs}+\bL^{\st}\wh{\bs}\ba,
\end{align}
\end{subequations}
or expanded into two separated systems:
\begin{subequations}
\begin{align}
&\bK\bM_{\iota}-\bM_{\iota}\bL=\br_{\iota}\bs_{\iota}^{\st},\\
&p\wt{\br}_{\iota}=p\br_{\iota}+(-1)^{\iota-1}\bK\br_{\iota},\quad q\wh{\br}_{\iota}=q\br_{\iota}+(-1)^{\iota-1}\bK\br_{\iota},\\
&p\wt{\bs}_{\iota}+(-1)^{\iota-1}\bL^{\st}\wt{\bs}_{\iota}=p\bs_{\iota},\quad q\wh{\bs}_{\iota}+(-1)^{\iota-1}\bL^{\st}\wh{\bs}_{\iota}=q\bs_{\iota},
\end{align}
\end{subequations}
with $\iota=1,2$. Then the master function
\begin{align}
\bS^{(i,j)}=\bs^{\st}\bL^{j}\bM^{-1}\bK^{i}\br
\end{align}
also satisfies relations \eqref{Sij1-ht}, and the functions $(U_{3},\mV)$ defined via this $\bS^{(i,j)}$ remain 
solutions to the dhGI models \eqref{dhGI1}, \eqref{dhGI2}, \eqref{dhGI3}, and \eqref{dhGI4}.
\end{prop}

\section{Remaining scalar dGI equations}\label{App-B}
The following three scalar dGI equations are obtained by applying the reduction $\nu=u_{3}^*$ to \eqref{dGI-2}, \eqref{dGI-3} and \eqref{dGI-4}, respectively:
\begin{subequations}
\label{d-gi-rest}
\begin{align}
& 2|p|^2(q(\wh{u}_{3}-\wh{\wt{u}}_{3})+q^*(\wt{u}_{3}-u_{3}))+2|q|^2(p(\wh{\wt{u}}_{3}-\wt{u}_{3})+p^*(u_{3}-\wh{u}_{3}))+(p^*q\wh{u}_{3}-pq^*\wt{u}_{3})\times \nn \\
&(q^*-q+p^*-p+\theta_{q}^{-1}\wh{u}_{3}u_{3}^*-\theta_{q}u_{3}\wh{u}_{3}^*+\theta_{p}^{-1}\wh{\wt{u}}_{3}\wh{u}_{3}^*-\theta_{p}\wh{u}_{3}\wh{\wt{u}}_{3}^*
+|\wh{\wt{u}}_{3}|^2-|u_{3}|^2-Q_{1}-\wh{P}_{1}) \nn \\
&+(p^*q^*u_{3}-pq\wh{\wt{u}}_{3})(q^*-q+p-p^*-\theta_{p}^{-1}\wh{\wt{u}}_{3}\wh{u}_{3}^*+\theta_{p}\wh{u}_{3}\wh{\wt{u}}_{3}^*+\theta_{q}^{-1}\wh{\wt{u}}_{3}\wt{u}_{3}^*
-\theta_{q}\wt{u}_{3}\wh{\wt{u}}_{3}^* \nn \\
& +|\wh{u}_{3}|^2-|\wt{u}_{3}|^2-\wt{Q}_{1}+\wh{P}_{1})=0, \\
& 2|p|^2(q(\wh{u}_{3}-\wh{\wt{u}}_{3})+q^*(\wt{u}_{3}-u_{3}))+2|q|^2(p(\wh{\wt{u}}_{3}-\wt{u}_{3})+p^*(u_{3}-\wh{u}_{3}))+(p^*q\wh{u}_{3}-pq^*\wt{u}_{3})\times \nn \\
& (q^*-q+p^*-p+\theta_{q}^{-1}\wh{\wt{u}}_{3}\wt{u}_{3}^*-\theta_{q}\wt{u}_{3}\wh{\wt{u}}_{3}^*+\theta_{p}^{-1}\wt{u}_{3}u_{3}^*-\theta_{p}u_{3}\wt{u}_{3}^*
+|\wh{\wt{u}}_{3}|^2-|u_{3}|^2-\wt{Q}_{1}-P_{1}) \nn \\
& +(p^*q^*u_{3}-pq\wh{\wt{u}}_{3})(q^*-q+p-p^*-\theta_{p}^{-1}\wt{u}_{3}u_{3}^*+\theta_{p}u_{3}\wt{u}_{3}^*+\theta_{q}^{-1}\wh{u}_{3}u_{3}^*
-\theta_{q}u_{3}\wh{u}_{3}^* \nn \\ 
& +|\wh{u}_{3}|^2-|\wt{u}_{3}|^2-Q_{1}+P_{1})=0, 
\end{align}
\begin{align}
& 2|p|^2(q(\wh{u}_{3}-\wh{\wt{u}}_{3})+q^*(\wt{u}_{3}-u_{3}))+2|q|^2(p(\wh{\wt{u}}_{3}-\wt{u}_{3})+p^*(u_{3}-\wh{u}_{3}))+(p^*q\wh{u}_{3}-pq^*\wt{u}_{3})\times \nn \\
&(q^*-q+p^*-p+\theta_{q}^{-1}\wh{\wt{u}}_{3}\wt{u}_{3}^*-\theta_{q}\wt{u}_{3}\wh{\wt{u}}_{3}^*+\theta_{p}^{-1}\wt{u}_{3}u_{3}^*-\theta_{p}u_{3}\wt{u}_{3}^*
+|\wh{\wt{u}}_{3}|^2-|u_{3}|^2-\wt{Q}_{1}-P_{1}) \nn \\ 
&+(p^*q^*u_{3}-pq\wh{\wt{u}}_{3})(q^*-q+p-p^*-\theta_{p}^{-1}\wh{\wt{u}}_{3}\wh{u}_{3}^*+\theta_{p}\wh{u}_{3}\wh{\wt{u}}_{3}^*
+\theta_{q}^{-1}\wh{\wt{u}}_{3}\wt{u}_{3}^*-\theta_{q}\wt{u}_{3}\wh{\wt{u}}_{3}^* \nn \\ 
& +|\wh{u}_{3}|^2-|\wt{u}_{3}|^2-\wt{Q}_{1}+\wh{P}_{1})=0,
\end{align}
\end{subequations}
where $P_{1}$ and $Q_{1}$ are defined as the equation \eqref{PQ-1}.

\section{Remaining scalar dhGI equations (local reduction)}\label{App-C}

The following three scalar dhGI equations are obtained by applying the local reduction $\mV=U_{3}^*$ to models \eqref{dhGI2}, \eqref{dhGI3} and \eqref{dhGI4}, respectively:
\begin{subequations}
\label{eq-hGI-rest}
\begin{align}
& (U_{3}(\wt{U}_{3}^*-\wh{U}_{3}^*)+U_{3}^*(\wh{U}_{3}-\wt{U}_{3})-|\wh{U}_{3}|^2+|\wt{U}_{3}|^2-\mathrm{i}P^{'}_{1}+\mathrm{i}Q_{1}^{'}) (\wh{\wt{U}}_{3}-U_{3}) \nn \\
& \qquad +(\wt{U}_{3}(\wh{\wt{U}}_{3}^*-U_{3}^*)+\wt{U}_{3}^*(U_{3}-\wh{\wt{U}}_{3})+|\wh{\wt{U}}_{3}|^2-|U_{3}|^2-\mathrm{i}P_{1}^{'}-\mathrm{i}\wt{Q}_{1}^{'}) (\wt{U}_{3}-\wh{U}_{3})=0,\\
&  (\wh{\wt{U}}_{3}(\wt{U}_{3}^*-\wh{U}_{3}^*)+\wh{\wt{U}}_{3}^*(\wh{U}_{3}-\wt{U}_{3})-|\wh{U}_{3}|^2+|\wt{U}_{3}|^2-\mathrm{i}\wh{P}_{1}^{'}+\mathrm{i}\wt{Q}_{1}^{'})
(\wh{\wt{U}}_{3}-U_{3}) \nn \\ 
& \qquad +(\wh{U}_{3}(\wh{\wt{U}}_{3}^*-U_{3}^*)+\wh{U}_{3}^*(U_{3}-\wh{\wt{U}}_{3})+|\wh{\wt{U}}_{3}|^2-|U_{3}|^2-\mathrm{i}\wh{P}_{1}^{'}-\mathrm{i}Q_{1}^{'}) (\wt{U}_{3}-\wh{U}_{3})=0, \\
& (\wh{\wt{U}}_{3}(\wt{U}_{3}^*-\wh{U}_{3}^*)+\wh{\wt{U}}_{3}^*(\wh{U}_{3}-\wt{U}_{3})-|\wh{U}_{3}|^2+|\wt{U}_{3}|^2-\mathrm{i}\wh{P}_{1}^{'}+\mathrm{i}\wt{Q}_{1}^{'}) (\wh{\wt{U}}_{3}-U_{3}) \nn \\
& \qquad +(\wt{U}_{3}(\wh{\wt{U}}_{3}^*-U_{3}^*)+\wt{U}_{3}^*(U_{3}-\wh{\wt{U}}_{3})+|\wh{\wt{U}}_{3}|^2-|U_{3}|^2-\mathrm{i}P_{1}^{'}-\mathrm{i}\wt{Q}_{1}^{'}) (\wt{U}_{3}-\wh{U}_{3})=0,
\end{align}
\end{subequations}
where $P^{'}_{1}$ and $Q^{'}_{1}$ are defined as the equation \eqref{P'Q'-1}.

\section{Remaining nonlocal dhGI equations}
\label{App-D}

The following three nonlocal dhGI equations are obtained from models $(N.2)$, $(N.3)$ and $(N.4)$ under the nonlocal reduction $\mV=\mathrm{i}U_{3,-1}^*$:
\begin{subequations}
\label{non-eq-hGI-rest}
\begin{align}
& \big(U_{3}(\wt{U}_{3,-1}^*-\wh{U}_{3,-1}^*)+U_{3,-1}^*(\wh{U}_{3}-\wt{U}_{3})-\wh{U}_{3}\wh{U}_{3,-1}^*+\wt{U}_{3}\wt{U}_{3,-1}^*-P^{'}_{2}+Q_{2}^{'}\big)
(\wh{\wt{U}}_{3}-U_{3}) \nn \\ 
& +\big(\wt{U}_{3}(\wh{\wt{U}}_{3,-1}^*-U_{3,-1}^*)+\wt{U}_{3,-1}^*(U_{3}-\wh{\wt{U}}_{3})+\wh{\wt{U}}_{3}\wh{\wt{U}}_{3,-1}^*-U_{3}U_{3,-1}^*
-P_{2}^{'}-\wt{Q}_{2}^{'}\big) (\wt{U}_{3}-\wh{U}_{3})=0, \\
& \big(\wh{\wt{U}}_{3}(\wt{U}_{3,-1}^*-\wh{U}_{3,-1}^*)+\wh{\wt{U}}_{3,-1}^*(\wh{U}_{3}-\wt{U}_{3})-\wh{U}_{3}\wh{U}_{3,-1}^*
+\wt{U}_{3}\wt{U}_{3,-1}^*-\wh{P}_{2}^{'}+\wt{Q}_{2}^{'}\big) (\wh{\wt{U}}_{3}-U_{3}) \nn \\
& + \big(\wh{U}_{3}(\wh{\wt{U}}_{3,-1}^*-U_{3,-1}^*)+\wh{U}_{3,-1}^*(U_{3}-\wh{\wt{U}}_{3})+\wh{\wt{U}}_{3}\wh{\wt{U}}_{3,-1}^*-U_{3}U_{3,-1}^*
-\wh{P}_{2}^{'}-Q_{2}^{'}\big) (\wt{U}_{3}-\wh{U}_{3})=0, \\
& \big(\wh{\wt{U}}_{3}(\wt{U}_{3,-1}^*-\wh{U}_{3,-1}^*)+\wh{\wt{U}}_{3,-1}^*(\wh{U}_{3}-\wt{U}_{3})-\wh{U}_{3}\wh{U}_{3,-1}^*+\wt{U}_{3}\wt{U}_{3,-1}^*-\wh{P}_{2}^{'}
+\wt{Q}_{2}^{'}\big) (\wh{\wt{U}}_{3}-U_{3}) \nn \\
& +\big(\wt{U}_{3}(\wh{\wt{U}}_{3,-1}^*-U_{3,-1}^*)+\wt{U}_{3,-1}^*(U_{3}-\wh{\wt{U}}_{3})+\wh{\wt{U}}_{3}\wh{\wt{U}}_{3,-1}^*-U_{3}U_{3,-1}^*
-P_{2}^{'}-\wt{Q}_{2}^{'}\big) (\wt{U}_{3}-\wh{U}_{3})=0,
\end{align}
\end{subequations}
where $P^{'}_{2}$ and $Q^{'}_{2}$ are defined as the equation \eqref{P'Q'-2}.

\end{appendix}

\end{document}